\documentclass[11pt]{article}
\usepackage{times}
\usepackage{epsfig}
\usepackage{amsmath,amsthm}
\usepackage{amsfonts}
\usepackage{amssymb}
\usepackage{amstext}
\usepackage{latexsym}
\usepackage{color}
\usepackage{ifthen}
\usepackage{multirow}
\usepackage{subfigure}
\usepackage{pstricks}
\usepackage{verbatim}

\newtheorem{lemma}{Lemma}
\newtheorem{theorem}{Theorem}

\newtheorem{remark}{Remark}

\def\appendix{\par
    \setcounter{section}{0}\setcounter{subsection}{0}
    \def\thesection{Appendix:}
    \def\thesubsection{A\arabic{section}.}
}

\newcommand{\F}{{\mathbb{F}}}
\newcommand{\Z}{{\mathbb{Z}}}

\newcommand{\N}{\mathcal{N}}

\newcommand{\al}{\alpha}

\newcommand{\eps}{\epsilon}
\newcommand{\snr}{\textsf{SNR}}
\newcommand{\inr}{\textsf{INR}}

\newcommand{\M}{\mathcal{M}}
\newcommand{\X}{\mathcal{X}}

\newcommand{\y}{\tilde{y}}
\newcommand{\yb}{\bar{y}}
\newcommand{\yh}{\hat{y}}
\newcommand{\x}{\hat{x}}
\newcommand{\xb}{\bar{x}}

\newcommand{\lf}{\lfloor}
\newcommand{\rf}{\rfloor}
\newcommand{\onen}{\frac{1}{N}}
\newcommand{\onenp}{\frac{1}{N'}}
\newcommand{\R}{\mathbb{R}}
\DeclareMathOperator{\E}{E}
\DeclareMathOperator{\Prob}{P}
\DeclareMathOperator{\sign}{sign}
\DeclareMathOperator{\supp}{supp}
\DeclareMathOperator{\fracp}{frac}
\newcommand{\nn}{\nonumber}
\newcommand{\half}{\frac{1}{2}}
\newcommand{\op}{\oplus}

\newcommand{\ve}{\varepsilon}
\newcommand{\CalC}{\mathcal{C}}
\newcommand{\C}{\mathbb{C}}
\newcommand{\CN}{\mathcal{CN}}
\newcommand{\D}{\mathcal{D}}
\newcommand{\De}{\Delta}
\newcommand{\A}{\tilde{A}}
\newcommand{\B}{\tilde{B}}
\newcommand{\U}{\tilde{U}}
\newcommand{\V}{\tilde{V}}

\begin{document}
\title{The Two-User Gaussian Interference Channel: A Deterministic View}
\author{Guy Bresler\thanks{Department of Electrical Engineering and Computer Sciences, University of California, 
 Berkeley, California, USA (email: gbresler@eecs.berkeley.edu).} \and David Tse  \thanks{David Tse is with the Department of Electrical Engineering and Computer Sciences, University of California, 
 Berkeley, California, USA (email: dtse@eecs.berkeley.edu).} \thanks{The research was supported by a Vodafone Fellowship and by the National Science Foundation under an ITR grant: the 3R's of Spectrum Management: Reuse, Reduce and Recycle. }} 
\date{}
\maketitle
\begin{abstract}
  This paper explores the two-user Gaussian interference channel through the lens of a natural deterministic channel model. The main result is that the deterministic channel uniformly approximates the Gaussian channel, the capacity regions differing by a universal constant. The problem of finding the capacity of the Gaussian channel to within a constant error is therefore reduced to that of finding the capacity of the far simpler deterministic channel. Thus, the paper provides an alternative derivation of the recent constant gap capacity characterization of Etkin, Tse, and Wang \cite{ETW07}. Additionally, the deterministic model gives significant insight towards the Gaussian channel.
\end{abstract}

\section{Introduction}




One of the longest outstanding problems in multiuser information theory is the capacity region of the two-user
Gaussian interference channel. This multiuser channel consists of two point-to-point links with
additive white Gaussian noise, interfering with each other through crosstalk (Figure
\ref{fig:model}).

\begin{figure}[htb]
\centerline{ \psfig{figure=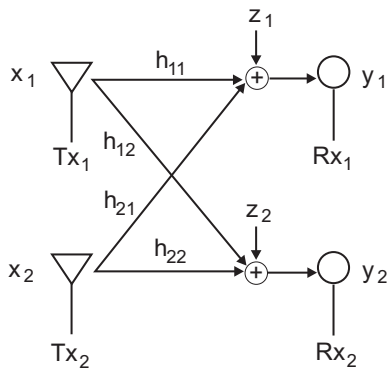,width=2in}} \caption{\small
Two-user Gaussian interference channel. } \label{fig:model}
\vspace{-.1in}
\end{figure}

Each transmitter has an independent message intended only for the corresponding receiver. 
The capacity region of this channel is the set of
all simultaneously achievable rate pairs $(R_1,R_2)$ in the two
interfering links, and characterizes the fundamental tradeoff
between the performance achievable in the links in the face of
interference. Unfortunately, the problem of characterizing this
region has been open for over thirty years. The capacity region is known in the {\em strong} interference case,
where each receiver has a better reception of the other user's
signal than the intended receiver \cite{HK81,C75}. The best known
strategy for the other cases is due to Han and Kobayashi
\cite{HK81}. This strategy is a natural one and involves splitting
the transmitted information of both users into two parts: private
information to be decoded only at own receiver and common
information that can be decoded at both receivers. By decoding the
common information, part of the interference can be canceled off,
while the remaining private information from the other user is
treated as noise. The Han-Kobayashi strategy allows arbitrary
splits of each user's transmit power into the private and common
information portions as well as time sharing between multiple such
splits. Unfortunately, the optimization among such myriads of
possibilities is not well-understood, and it is also not
clear how close to capacity can such a scheme get and whether
there will be other strategies that can do significantly better.

Significant progress on this problem has been made recently. In \cite{ETW07}, it was shown that a very simple Han-Kobayashi type scheme can in fact achieve
rates within $1$ bits/s/Hz of the capacity of the channel for {\em
all} values of the channel parameters. That is, this scheme can
achieve the rate pair $(R_1-1,R_2-1)$ for any $(R_1,R_2)$ in the
interference channel capacity region. This result is particularly
relevant in the high signal-to-noise ratio (SNR) regime, where the
achievable rates are high and grow unbounded as the noise level goes
to zero. The high
SNR regime is the interference-limited scenario: when the noise is
small, interference from one link will have a significant impact on
the performance of the other. 
Progress has also been made towards finding the exact capacity region; by extending one of the converse arguments in \cite{ETW07}, the authors of \cite{SKC07} and \cite{AV08} show that treating interference as noise is sum-rate optimal when the interference is sufficiently weak.

The purpose of the present paper is to show that the high SNR behavior of the Gaussian interference channel characterized in \cite{ETW07} can in fact be fully captured by a natural underlying {\em deterministic} interference channel. This type of deterministic channel model was first proposed by \cite{ADT07} in the analysis of Gaussian relay networks. Applying this model to the interference scenario, we show that the capacity of the resulting deterministic interference channel is \emph{the same}---to within a constant number of bits---as the corresponding Gaussian interference channel. 
Combined with the capacity result for the two-user deterministic interference channel, the paper therefore provides an alternative derivation of the constant gap result of \cite{ETW07} (albeit with a larger gap). 

Because of the simplicity of the deterministic channel model, it provides a lot of insight to the structure of the various near-optimal schemes for the Gaussian interference channel in the different parameter ranges. 
Where certain approximate statements and intuitions can be made regarding the Gaussian interference channel,
these statements are made precise in the deterministic setting. The near-optimality for the Gaussian channel of the simple Han-Kobayashi scheme as shown in \cite{ETW07} is made transparent in the deterministic channel:
the derivation of the achievable strategy is completed in a series of steps, each shown to be without loss of optimality. As an added benefit, the relatively complicated genie-aided converse arguments are avoided. 

The close connection between the deterministic and Gaussian  channels, as demonstrated in the example of the two-user interference channel discussed in this paper, suggests a new general approach to attack multiuser information theory problems. Given a Gaussian network, one can attempt to \emph{reduce} the Gaussian problem to a deterministic one by proving a constant gap between the capacity regions of the two models. It then remains only to find the capacity of the presumably simpler deterministic channel. 
In \cite{BPT07}, the less direct approach of transferring proof techniques from the deterministic to Gaussian channel has been used successfully in approximating the capacity of the Gaussian many-to-one interference channel, where there is an arbitrary number of users but interference only happens at a single receiver.
The approach used in \cite{BPT07} is therefore taken a step further in this work. 


\section{Generalized Degrees of Freedom and Deterministic Model for the MAC}\label{sec:GenDFandDetMAC}

\subsection{Generalized Degrees of Freedom}\label{subsec:GenDF}
Before the one-bit gap result \cite{ETW07}, very little was known about the structure of the capacity region of the two-user Gaussian interference channel. The investigation of the \emph{generalized degrees of freedom}, a concept introduced in \cite{ETW07}, provided the first and crucial insight into the problem. In this section we motivate this idea through the MAC, as well as provide a more abstract look into what makes the generalized degrees of freedom so useful towards understanding the Gaussian interference channel. 

Let us start with the point-to-point AWGN channel.
The output is equal to $$y=\sqrt \snr x+z\,,$$ where $z\in\CN(0,1)$ and the input satisfies an average power constraint $$\onen \sum_{k=1}^N \E[x_k^2]\leq 1\,.$$ The capacity is equal to $$\CalC(\snr)=\log(1+\snr)\,.$$ In an attempt to capture the rough behavior of the capacity, one may calculate the limit 
\begin{equation}
  \label{e:p2pDF}
  \lim_{\snr\to\infty} \frac{\CalC(\snr)}{\log \snr}=1\,.
\end{equation}
The limit in \eqref{e:p2pDF}, the so-called degrees of freedom of the channel, measures how the capacity scales with SNR. The degrees of freedom is thus a rough measure of capacity, with unit equal to a single AWGN channel with appropriate SNR.  

We now attempt a similar understanding for the MAC.
The channel output is $$y= h_1 x_1+h_2 x_2 +z_1\,$$ where $h_1,h_2\in\C$, $z_i\sim\CN(0,1)$, and each input satisfies an average power constraint $$\onen \sum_{k=1}^N \E[x_{i,k}^2]\leq P_i\,,\quad i=1,2\, .$$
The channel is parameterized by the signal-to-noise ratios $\snr_1=P_1 |h_1|^2$ and $\snr_2=P_2|h_2|^2$, and we assume without loss of generality that $\snr_1\geq \snr_2$. The capacity region of the MAC is (see Figure~\ref{fig:MACregion}):
\begin{equation}\begin{split}\label{e:MACsimpleExpression}
  R_1 &\leq \log (1+\snr_1) \\
  R_2 &\leq \log (1+\snr_2) \\
  R_1+R_2 &\leq \log (1+\snr_1+\snr_2)\,.
  \end{split}
\end{equation}

Seeking simplification, a reasonable strategy is to attempt to compute a limit similar to \eqref{e:p2pDF}. However, there is not a clear choice of limit: the point-to-point channel had only one parameter and thus no ambiguity arose, but in the MAC there are two parameters $\snr_1$ and $\snr_2$ and therefore many ways of taking limits. 
Let $\CalC(h_1,h_2,P)$ denote the capacity region of the MAC \eqref{e:MACsimpleExpression} with channel gains $h_1,h_2$ and power constraint $P$ for both users. One standard way of taking the limit of the region is to let the power constraint $P$ tend to infinity, scaling by $\log P$:
$$\lim_{P\to \infty} \frac{\CalC(h_1,h_2,P)}{\log P}\,.$$
Calculating the limit, one finds that the resulting region (see Figure~\ref{fig:MAC_ClassicalDF})
\begin{equation} \label{e:MAC_Classical_df}
  \begin{split}
    d_1&\leq 1 \\
    d_2&\leq 1 \\
    d_1+d_2&\leq 1
  \end{split}
\end{equation}
is altogether independent of the channel gains. 
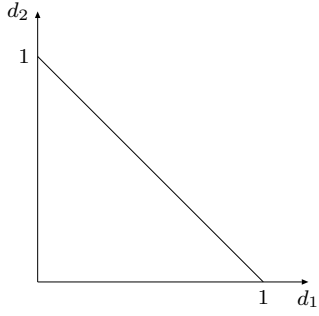
\begin{figure}
\begin{centering}
\psset{unit=1.2mm,linewidth=.3pt,arrowlength=1.2,
arrowinset=0,labelsep=3pt}
\begin{center}
\begin{pspicture}(0,2)(30,32)
\psline{->}(0,0)(30,0)
\psline{->}(0,0)(0,30)

\uput[l](0,30){\scriptsize $d_2$}
\uput[d](30,0){\scriptsize $d_1$}

\psline(0,25)(25,0)
\uput[d](25,0){\scriptsize $1$}
\uput[l](0,25){\scriptsize $1$}

\end{pspicture}
\end{center} 
\caption{The classical degrees of freedom region for the MAC.}
\label{fig:MAC_ClassicalDF}
\end{centering}
\end{figure}
More troubling, the limiting region \eqref{e:MAC_Classical_df} is misleading from an operational viewpoint. The region seems to suggest that for high transmit powers, the optimal scheme is time-sharing between the two rate points in which only one user transmits at a time. But this is far from the truth, as a corner point of the capacity region has an arbitrarily greater sum-rate as channel parameters are varied, for each fixed power constraint. This limit, therefore, does not reveal any dynamic range between users, a quality that is relevant at finite SNR.

\begin{figure}
\begin{centering}
\psset{unit=1.2mm,linewidth=.3pt,arrowlength=1.2,
arrowinset=0,labelsep=4pt}
\begin{center}
\begin{pspicture}(-5,-1)(40,35)
\psline{->}(0,0)(40,0)
\psline{->}(0,0)(0,30)

\uput[l](0,30){\scriptsize $r_2$}
\uput[d](40,0){\scriptsize $r_1$}

\psline(0,20)(14,20)(32,2)(32,0)

\psline[linestyle=dashed,dash=2pt 2pt](0,18)(12,18)(30,0)


\psline(34,-4)(32,-.5)
\uput[d](32,-3.2){\scriptsize $\log(1+ \snr_1)$}
\psline(29.5,-.5)(26,-2.5)
\rput(21,-2.5){\scriptsize $\log \snr_1$}

\uput[l](0,21){\scriptsize $\log(1+ \snr_2)$}
\uput[l](0,17){\scriptsize $\log \snr_2$}


\psline(20,14.5)(22,18)
\rput(24,19.5){\scriptsize MAC capacity}

\psline(16,13.5)(14,10)
\uput[d](13,10.5){\scriptsize approx.}
\uput[d](13,9){\scriptsize capacity}

\end{pspicture}
\end{center} 
\caption{The solid line shows the MAC capacity region. The dashed line shows the approximate region as given in \eqref{e:MACapprox}, and is within one bit per user of the capacity region.}
\label{fig:MACregion}
\end{centering}
\end{figure}
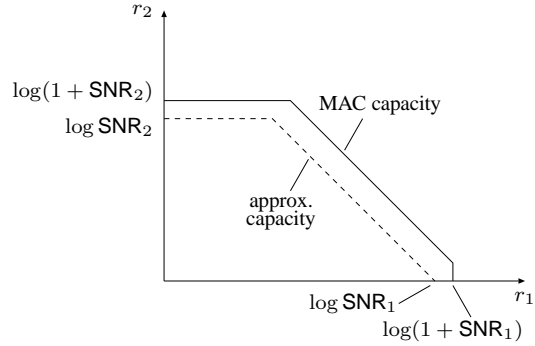

A closer look at the capacity region itself leads to a different limit.
Notice that the capacity region can be approximated to within one bit per user as (see Figure~\ref{fig:MACregion})
\begin{equation}\begin{split} \label{e:MACapprox}
  R_1 &\leq \log (1+\snr_1)\approx \log\snr_1 \\
  R_2 &\leq \log (1+\snr_2) \approx \log \snr_2\\
  R_1+R_2 &\leq \log (1+\snr_1+\snr_2)\approx \log \snr_1\,.
\end{split}\end{equation}
In order to roughly preserve the shape of the capacity region in the limit, equation \eqref{e:MACapprox} suggests to fix the relationship between the two individual rate constraints, i.e. $$\log \snr_2 = \alpha \log \snr_1\,.$$
In other words, the ratio of SNRs is fixed in the dB scale.
This is precisely the generalized degrees of freedom limit,
$$
\D(\al):=\lim_{{\snr\to \infty} } \frac{ \CalC(\snr,\snr^\al)}{\log(\snr)}
\,,$$
where $\CalC(\snr_1,\snr_2)$ denotes the capacity region of the MAC with signal-to-noise ratios $\snr_1,\snr_2$. 
The resulting region (Figure~\ref{fig:MACgenDF}) is
\begin{equation}\begin{split}\label{e:MACgenDF}
  d_1&\leq 1 \\
  d_2&\leq \al \\
  d_1+d_2&\leq 1\,.
\end{split}\end{equation}

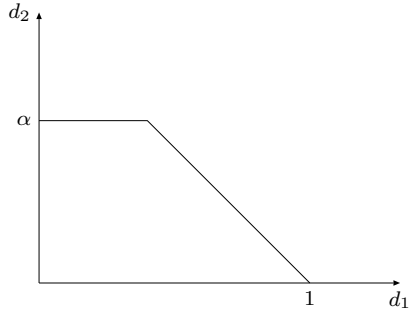
\begin{figure}
\begin{centering}
\psset{unit=1.2mm,linewidth=.3pt,arrowlength=1.2,
arrowinset=0,labelsep=3pt}
\begin{center}
\begin{pspicture}(-5,2)(40,32)
\psline{->}(0,0)(40,0)
\psline{->}(0,0)(0,30)

\uput[l](0,30){\scriptsize $d_2$}
\uput[d](40,0){\scriptsize $d_1$}

\psline(0,18)(12,18)(30,0)


\uput[l](0,18){\scriptsize $\al$}
\uput[d](30,0){\scriptsize $1$}

\end{pspicture}
\end{center} 
\caption{The MAC generalized degrees of freedom region. The region is exactly the same as the approximate region in Figure~\ref{fig:MACregion}, normalized by $\log \snr_1$.}
\label{fig:MACgenDF}
\end{centering}
\end{figure}

Qualitatively, the generalized degrees of freedom limit preserves the dynamic range feature of the finite-SNR channel. However, a more precise statement is true as well: because the approximation to the region \eqref{e:MACapprox} is to within one bit, independent of the channel gains, it follows that the degrees of freedom region itself, when scaled by $\log \snr_1$, is within one bit of the true region. Thus, varying $\alpha$, the limiting regions~\eqref{e:MACgenDF} \emph{uniformly cover} the entire collection of finite signal-to-noise ratio channels. In other words, to find the approximate capacity of any MAC with (finite) signal-to-noise ratios $\snr_1,\snr_2$, one simply needs to compute the generalized degrees of freedom limit for the value $\alpha=\frac{\log \snr_2}{\log \snr_1}$.



\begin{figure}
\begin{centering}
\psset{unit=1.2mm,linewidth=.3pt,arrowlength=1.2,
arrowinset=0,labelsep=3pt}
\begin{center}
\begin{pspicture}(0,2)(40,25)
\psline{->}(0,0)(40,0)
\psline{->}(0,0)(0,20)

\uput[l](0,20){\scriptsize $\snr_2$}
\uput[d](40,0){\scriptsize $\snr_1$}

\pscurve(3,0)(10,6)(18,4)(25,8)(40,11)
\uput[u](10,6){$f$}

\end{pspicture}
\end{center} 
\caption{An example limit path in the $(\snr_1,\snr_2)$ plane.}
\label{fig:limitPath}
\end{centering}
\end{figure}
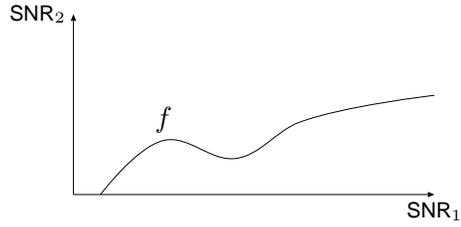

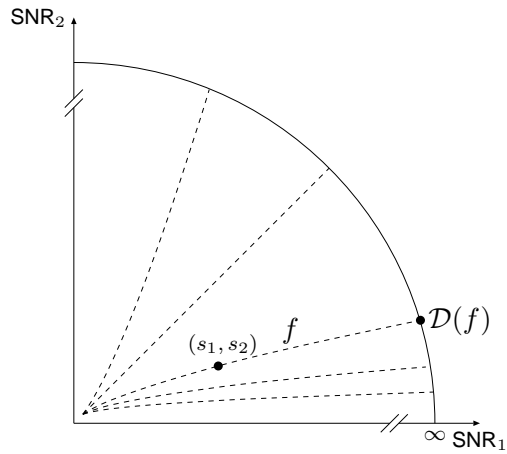
\begin{figure}
\begin{centering}
\psset{unit=1.2mm,linewidth=.3pt,arrowlength=1.2,
arrowinset=0,labelsep=3pt}
\begin{center}
\begin{pspicture}(0,2)(45,45)
\psline(0,0)(35,0)
\psline(0,0)(0,35)
\psline(34,-1)(36,1)
\psline(35,-1)(37,1)
\psline(-1,34)(1,36)
\psline(-1,35)(1,37)
\psline{->}(0,36)(0,45)
\psline{->}(36,0)(45,0)

\uput[l](0,45){\scriptsize $\snr_2$}
\uput[d](45,0){\scriptsize $\snr_1$}

\uput[d](40,0){\scriptsize $\infty$}

\psarc(0,0){40}{0}{90}

\pscurve[linestyle=dashed,dash=2pt 2pt](1,1)(2,1.41)(4,2)(16,4)(25,5)(32,5.66)(39,6.24)
\pscurve[linestyle=dashed,dash=2pt 2pt](1,1)(2,1.26)(4,1.59)(16,2.52)(25,2.92)(32,3.175)(39.8,3.46)
\pscurve[linestyle=dashed,dash=2pt 2pt](1,1)(2,1.59)(4,2.52)(16,6.35)(25,8.55)(32,10.08)(38.4,11.42)
\psline[linestyle=dashed,dash=2pt 2pt](1,1)(28.28,28.28)
\pscurve[linestyle=dashed,dash=2pt 2pt](1,1)(2,2.52)(4,6.35)(8,16)(14,33.74)(14.7,36)(15,37)

\pscircle*(16,6.35){.5}
\pscircle*(38.4,11.42){.5}

\uput[u](16.5,6.7){\scriptsize $(s_1,s_2)$}
\uput[u](24,8){$f$} 
\uput[r](38.4,11.42){$\D(f)$}

\end{pspicture}
\end{center} 
\caption{The figure illustrates the notion of a limit region uniformly approximating the capacity region. Suppose the capacity, scaled by $\log \snr_1$, is constant along the limit paths. The dashed lines show several example limit paths. Then, to find the capacity region at any point $(s_1,s_2)$ in the $(\snr_1,\snr_2)$ plane, one may simply follow the path (denoted by $f$) to the infinite arc, resulting in $\D(f)$.}
\label{fig:uniformApproximation}
\end{centering}
\end{figure}

In the MAC, we observed that the generalized degrees of freedom limit correctly expresses the finite-SNR behavior. We now reflect on what properties, more abstractly, constitute a useful limit.
Visually, a limit corresponds to a choice of path, 
$(\snr,f(\snr))$ in the $(\snr_1,\snr_2)$ plane (Figure~\ref{fig:limitPath}). 
Thus, a first requirement is to choose a function $f$ such that the limit exists: 
\begin{equation}\label{e:dfScaling}\lim_{\snr \to\infty} \frac{\CalC(\snr,f(\snr))}{\log \snr}=\D(f)\,.
\end{equation}

Although many trajectories are possible, if the goal is a better understanding of the capacity region for \emph{finite} power-to-noise ratios, some limit paths are better than others.
Suppose, for example, that it was possible to choose $f$ such that 
\begin{equation}\label{e:constantTrajectory}
\frac{\CalC(\snr,f(\snr))}{\log \snr}= \text{constant}(f)\end{equation} for the entire range $\snr>0$. In words, the scaled capacity in~\eqref{e:constantTrajectory} is constant along the path $f$. In this case, the problem of finding the limit \eqref{e:dfScaling} is precisely the same as that of finding the capacity region for each point along the entire trajectory! Moreover, if after computing the limit one could vary $f$ so as to cover all points $(\snr,\inr)$, the problem of finding the capacity of the channel is completely solved.  

Figure~\ref{fig:uniformApproximation} further explains this idea. We consider the scaled (by $\log \snr$) capacity region. After taking a limit, one has the scaled capacity region at each point on an arc of infinite radius. Now, upon choosing an arbitrary point $(s_1,s_2)$ in the $(\snr_1,\snr_2)$ plane, a good limit should allow to deduce, from the scaled capacities on the infinite-radius arc, the (approximate) scaled capacity at $(s_1,s_2)$.
Hence the significance of \eqref{e:constantTrajectory}, which allows to equate the scaled capacity at finite SNRs with the limiting regions: if condition \eqref{e:constantTrajectory} is satisfied, one may simply choose the path $f$ containing the point $(s_1,s_2)$, which gives $$ \CalC(s_1,s_2)=\CalC(s_1,f(s_1))=\log s_1 \cdot \D(f)\,.$$

For the MAC, the set of trajectories defining the generalized degrees of freedom limit satisfies \eqref{e:constantTrajectory} to within a universal constant, independent of $\snr$.
The generalized degrees of freedom of the Gaussian MAC \eqref{e:MACgenDF} is the limit \eqref{e:dfScaling} along the path $$f(s)=s^\al\,.$$
The generalized degrees of freedom of the MAC is intimately connected to, and captured by, a certain deterministic channel model. In fact, the capacity region of the deterministic channel is, when properly scaled, \emph{equal} to the generalized degrees of freedom region. Equivalently, the deterministic channel satisfies \eqref{e:constantTrajectory} exactly.


\subsection{Deterministic Channel}\label{subsec:MACDeterministicChannel}


In this section we introduce a deterministic channel model analogous to the Gaussian channel. This channel was first introduced in \cite{ADT07}. 
We begin by describing the deterministic channel model for the point-to-point AWGN channel, and then the two-user multiple-access channel. After understanding these examples, we present the deterministic interference channel.

Consider first the model for the point-to-point channel (see Figure~\ref{fig:p2pDeterministic}). The real-valued channel input is written in base 2; the signal---a vector of bits---is interpreted as occupying a succession of levels:
$$x=0.b_1 b_2 b_3 b_4 b_5\dots\,.$$
The most significant bit coincides with the highest level, the least significant bit with the lowest level. The levels attempt to capture the notion of \emph{signal scale}; a level corresponds to a unit of power in the Gaussian channel, measured on the dB scale. Noise is modeled in the deterministic channel by truncation. Bits of smaller order than the noise are lost. The channel may be written as $$y=\lf 2^n x\rf\,,$$ with the correspondence $n=\lf \log \snr \rf$.
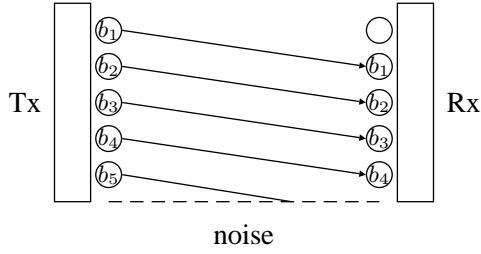
\begin{figure}
\begin{centering}
\psset{unit=1.2mm,arrowlength=1,arrowinset=0,linewidth=.5pt}
\begin{center}
\begin{pspicture}(-10,3)(40,24)

\pscircle(0,4){1.5}
\pscircle(0,8){1.5}
\pscircle(0,12){1.5}
\pscircle(0,16){1.5}
\pscircle(0,20){1.5}

\rput(0,4){\footnotesize $b_5$}
\rput(0,8){\footnotesize $b_4$}
\rput(0,12){\footnotesize $b_3$}
\rput(0,16){\footnotesize $b_2$}
\rput(0,20){\footnotesize $b_1$}

\psline{->}(1.5,16)(28.5,12)
\psline{->}(1.5,8)(28.5,4)
\psline{->}(1.5,12)(28.5,8)
\psline{->}(1.5,20)(28.5,16)

\pspolygon(-2,1)(-6,1)(-6,23)(-2,23)

\rput(30,0){
\pscircle(0,4){1.5}
\pscircle(0,8){1.5}
\pscircle(0,12){1.5}
\pscircle(0,16){1.5}
\pscircle(0,20){1.5}
\pspolygon(2,1)(6,1)(6,23)(2,23)
\rput(0,-4){\rput(0,8){\footnotesize $b_4$}
\rput(0,12){\footnotesize $b_3$}
\rput(0,16){\footnotesize $b_2$}
\rput(0,20){\footnotesize $b_1$}}
}

\psline[linestyle=dashed](0,1)(30,1)
\psline(1.5,4)(20.25,1)

\uput[l](-6,12){Tx}
\uput[r](36,12){Rx}
\uput[d](15,0){noise}
\end{pspicture}
\end{center}
\caption{The deterministic model for the point-to-point Gaussian channel. Each bit of the input occupies a signal level. Bits of lower significance are lost due to noise.}
\label{fig:p2pDeterministic}
\end{centering}
\end{figure}

The deterministic multiple-access channel is constructed similarly to the point-to-point channel (Figure~\ref{fig:deterministicMAC}), with $n_1$ and $n_2$ bits received above the noise level from users $1$ and $2$, respectively. To model the superposition of signals at the receiver, the bits received on each level are added {\em modulo two}. Addition modulo two, rather than normal integer addition, is chosen to make the model more tractable. As a result, the levels do not interact with one another.

If the inputs $x_i(t)$ are written in binary, the channel output can be written as
\begin{equation}\label{e:deterministicMACintegerpart}
  y=\lfloor 2^{n_{1}}x_1\rf \op \lf 2^{n_{2}}x_2\rfloor 
  \, ,
\end{equation}
where addition is performed on each bit (modulo two) and $\lfloor\, \cdot \, \rfloor$ is the integer-part function.
The channel can be written in an alternative form, which we will not use in the present paper but leads to a slightly different interpretation. The input and output are $x_1,x_2,y \in \F_2^q$, where $q=\max(n_1,n_2)$. The signal from transmitter $i $ is scaled by a nonnegative integer gain $2^{n_{i}} $ (equivalently, the input column vector is shifted up by $n_{i}$). The channel output is given by \begin{equation}y= \mathbf{S}^{q-n_{1}}x_1\op \mathbf{S}^{q-n_{2}}x_2,\end{equation} where summation and multiplication are in $\F_2$ and $\mathbf{S}$ is a $q\times q$ shift matrix,
\begin{equation}\label{e:shiftMatrix}
\mathbf{S}=\left(
\begin{matrix} 0 & 0 & 0 & \cdots & 0 \cr 1 & 0 & 0 & \cdots & 0 \cr 0 & 1 & 0 & \cdots & 0 \cr \vdots & & & \ddots & \vdots \cr 0 & \cdots & 0 & 1 & 0 \end{matrix}
\right).
\end{equation}

\begin{figure}
\begin{centering}
\psset{unit=.8mm,linewidth=.3pt,arrowlength=1.2,arrowinset=0,labelsep=5pt}
\begin{center}
\begin{pspicture}(-1,18)(37,68)

\rput(0,14){\pscircle(0,0){1.5}
\pscircle(0,4){1.5}
\pscircle(0,8){1.5}
\pscircle(0,12){1.5}
\pscircle(0,16){1.5}
\pspolygon(-2,-4)(-6,-4)(-6,20)(-2,20)

\uput[l](-5,8){$ {\text{Tx}_1}$}
}

\psline{->}(1.5,30)(28.5,34)
\psline{->}(1.5,26)(28.5,30)

\uput[r](35,38){$ \text{Rx}$}


\rput(0,46){
\pscircle(0,0){1.5}
\pscircle(0,4){1.5}
\pscircle(0,8){1.5}
\pscircle(0,12){1.5}
\pscircle(0,16){1.5}
\pspolygon(-2,-4)(-6,-4)(-6,20)(-2,20)
}
\psline{->}(1.5,46)(28.5,30)
\psline{->}(1.5,50)(28.5,34)
\psline{->}(1.5,54)(28.5,38)
\psline{->}(1.5,58)(28.5,42)
\psline{->}(1.5,62)(28.5,46)

\uput[l](-5,54){$\small \text{Tx}_2$}

\rput(30,30){
\pscircle(0,0){1.5}
\pscircle(0,4){1.5}
\pscircle(0,8){1.5}
\pscircle(0,12){1.5}
\pscircle(0,16){1.5}
\pspolygon(2,-4)(6,-4)(6,20)(2,20)
}

\end{pspicture}
\end{center} 
\caption{The deterministic model for the Gaussian multiple-access channel. Incoming bits on the same level are added modulo two at the receiver.}
\label{fig:deterministicMAC}
\end{centering}
\end{figure}
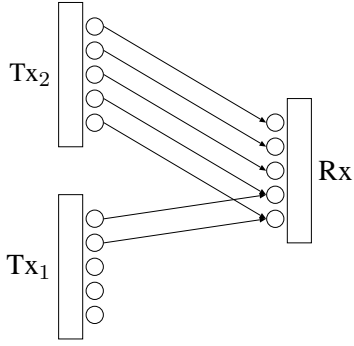

The capacity region of the deterministic MAC is 
\begin{equation}
  \begin{split}
    r_1 &\leq n_1 \\
    r_2 &\leq n_2 \\
    r_1 +r_2 &\leq \max(n_1,n_2)\,.
  \end{split}
\end{equation}
Comparing with \eqref{e:MACapprox}, we make the correspondence $$n_1=\lf \log \snr_1\rf,\quad n_2=\lf \log \snr_2\rf\,.$$
Evidently, the capacity region of the deterministic MAC is constant when normalized by $n_1$ and the ratio $\al=\frac{n_1}{n_2}$ is held fixed. Thus, the deterministic MAC satisfies \eqref{e:constantTrajectory} exactly when the gains are integer-valued; the normalized capacity along any point in the limit path is equal to the degrees of freedom of the deterministic MAC, which is in turn equal to the degrees of freedom of the Gaussian MAC.

\section{Deterministic Interference Channel}\label{sec:deterministicIC}
In Section~\ref{sec:GenDFandDetMAC} we motivated the generalized degrees of freedom limit and saw how it led to a simple deterministic model. The generalized degrees of freedom, and the equivalent deterministic model, was seen to uniformly approximate the MAC. With this success in explaining the MAC, a logical next step is to apply the deterministic model to the Gaussian interference channel. 

The Gaussian interference channel is given by 
\begin{align*}
  y_1&=h_{11}x_1+h_{12}x_2+z_1\\
  y_2&=h_{21}x_1+h_{22}x_2+z_2\,,
\end{align*}
where $z_i\sim \CN(0,1)$ and the channel inputs satisfy an average power constraint
$$\onen \sum_{k=1}^N \E[x_{i,k}^2]\leq P_i,\quad i=1,2\,.$$
The channel is parameterized by the power-to-noise ratios $\snr_1=|h_{11}|^2 P_1$, $\snr_2=|h_{22}|^2 P_2$, $\inr_1=|h_{21}|^2 P_1$, $\inr_2=|h_{12}|^2 P_2$.

We proceed with the deterministic interference channel model (Figure~\ref{fig:ICrectangles2}). Note that the model is completely determined by the model for the MAC.
There are two transmitter-receiver pairs (links), and as in the Gaussian case, each transmitter wants to communicate only with its corresponding receiver.  The signal from transmitter $j $, as observed at receiver $i$, is scaled by a nonnegative integer gain $2^{n_{ij}} $ (equivalently, the input column vector is shifted up by $n_{ij}$). At each time $t $, the input and output, respectively, at link $i$ are $x_i(t),y_i(t)\in \F_2^q$, where $q=\max_{ij}n_{ij}$.

The channel output at receiver $i$ is given by \begin{equation}y_i(t)= \mathbf{S}^{q-n_{i1}}x_1(t)\op \mathbf{S}^{q-n_{i2}}x_2(t),\end{equation} where summation and multiplication are in $\F_2$ and $\mathbf{S}$ is defined in \eqref{e:shiftMatrix}.

If the inputs $x_i$ are written in binary, the channel can equivalently be written as
\begin{align*}
  y_1&=\lfloor 2^{n_{11}}x_1\rf \op \lf 2^{n_{12}}x_2\rfloor \\
  y_2&=\lfloor 2^{n_{21}} x_1 \rf \op \lf 2^{n_{22}}x_2 \rfloor \, ,
\end{align*}
where addition is performed on each bit (modulo two) and $\lfloor\, \cdot \, \rfloor$ is the integer-part function. We will use the latter representation in this paper.

In the analysis of the deterministic interference channel, it will be helpful to consult a different style of figure. The left-hand side of Figure~\ref{fig:ICrectangles2} depicts a deterministic interference channel, and the right-hand side shows only the perspective of each receiver. Each incoming signal is shown as a column vector, with the highest element corresponding to the most significant bit and the portion below the noise level truncated. The observed signal at each receiver is the modulo 2 sum of the elements on each level. In the sequel, the dashed lines indicating the position of each entry of the vector will be omitted.
\begin{figure}
\begin{centering}
\psset{unit=.72mm,linewidth=.3pt,arrowlength=1.2,
arrowinset=0,labelsep=1.5pt}
\begin{center}
\begin{pspicture}(16,37)(80,75)

\rput(0,26){
\pscircle(0,4){1.5}
\pscircle(0,8){1.5}
\pscircle(0,12){1.5}
\pscircle(0,16){1.5}
\psline{->}(1.5,16)(24.5,16)
\psline{->}(1.5,8)(24.5,8)
\psline{->}(1.5,12)(24.5,12)
\psline{->}(1.5,4)(24.5,4)
\pspolygon(-2,0)(-6,0)(-6,20)(-2,20)
}

\uput[l](-6,36){$\scriptsize  \text{Tx}_2$}
\uput[r](32,36){$\scriptsize  \text{Rx}_2$}

\psline{->}(1.5,38)(24.5,56)
\psline{->}(1.5,42)(24.5,60)

\psline{->}(1.5,68)(24.5,30)

\rput(0,52){\pscircle(0,4){1.5}
\pscircle(0,8){1.5}
\pscircle(0,12){1.5}
\pscircle(0,16){1.5}
\psline{->}(1.5,16)(24.5,12)
\psline{->}(1.5,8)(24.5,4)
\psline{->}(1.5,12)(24.5,8)
\pspolygon(-2,0)(-6,0)(-6,20)(-2,20)
}

\uput[l](-6,62){$\scriptsize  \text{Tx}_1$}
\uput[r](32,62){$\scriptsize  \text{Rx}_1$}

\rput(26,26){
\pscircle(0,4){1.5}
\pscircle(0,8){1.5}
\pscircle(0,12){1.5}
\pscircle(0,16){1.5}
\pspolygon(2,0)(6,0)(6,20)(2,20)
}

\rput(26,52){
\pscircle(0,4){1.5}
\pscircle(0,8){1.5}
\pscircle(0,12){1.5}
\pscircle(0,16){1.5}
\pspolygon(2,0)(6,0)(6,20)(2,20)}


\rput(50,30){\psline[linestyle=dashed,dash=3pt 2pt](-1,0)(21,0)
\psline[linestyle=dashed,dash=3pt 2pt](34,0)(56,0)
\psline(1,0)(1,21)(7,21)(7,0) \rput(4,-2){\scriptsize 1}
\psline(13,0)(13,14)(19,14)(19,0)\rput(16,-2){\scriptsize 2}

\psline(36,0)(36,7)(42,7)(42,0)\rput(39,-2){\scriptsize 1}
\psline(48,0)(48,28)(54,28)(54,0) \rput(51,-2){\scriptsize 2}

\uput[r](7,21){\scriptsize  $n_{11}$} 
\uput[r](19,14){\scriptsize  $n_{12}$}

\uput[l](36.5,7){\scriptsize  $n_{21}$}
\uput[l](48,28){\scriptsize  $n_{22}$}

\psline[linestyle=dashed,dash=1pt 1pt](1,7)(7,7)
\psline[linestyle=dashed,dash=1pt 1pt](1,14)(7,14)

\psline[linestyle=dashed,dash=1pt 1pt](13,7)(19,7)

\psline[linestyle=dashed,dash=1pt 1pt](48,7)(54,7)
\psline[linestyle=dashed,dash=1pt 1pt](48,14)(54,14)
\psline[linestyle=dashed,dash=1pt 1pt](48,21)(54,21)

\rput(10,36){$\text{Rx}_1$}
\rput(45,36){$\text{Rx}_2$}



}
\end{pspicture}
\end{center} 
\caption{At left is a deterministic interference channel. The more compact figure at right shows only the signals as observed at the receivers.}
\label{fig:ICrectangles2}
\end{centering}
\end{figure}
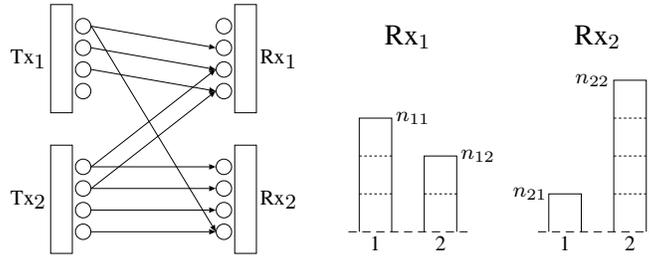

Just as in the discussion of the MAC, the deterministic interference channel {\em uniformly approximates} the Gaussian channel. In finding the capacity of the Gaussian interference channel to within a constant number of bits, it therefore suffices to find the capacity of the far simpler deterministic channel. 
\begin{theorem}\label{t:DeterministicApproximation}
The capacity of the two-user Gaussian interference channel with signal and interference
to noise ratios $\snr_1,\snr_2,\inr_1,\inr_2$ is within $42$ bits per
user of the capacity of a deterministic interference channel with
gains $2^{n_{11}}=2^{\lfloor \log \snr_1\rfloor }$, $2^{n_{12}}=2^{\lfloor \log \inr_2 \rfloor}$, $2^{n_{21}}=2^{\lfloor \log \inr_1\rfloor}$, and $2^{n_{22}}=2^{\lfloor \log \snr_2\rfloor}$.
\end{theorem}
\begin{proof}
  The capacity of the two-user Gaussian interference channel has been characterized to within one bit by Etkin, Tse, and Wang \cite{ETW07}; thus, we could prove the theorem by following the approach used for the MAC in Section~\ref{sec:GenDFandDetMAC}, comparing the capacity regions of the deterministic and Gaussian channels. We instead choose to prove the Theorem with no a priori knowledge of the result for the Gaussian channel. This approach provides insight into the deep connection between the deterministic and Gaussian channels, and also gives an alternative derivation of the constant-bit characterization of \cite{ETW07} (but with a significantly larger gap). The proof is deferred to the appendix.
\end{proof}

Theorem~\ref{t:DeterministicApproximation} gives as a corollary that the generalized degrees of freedom of the two-user Gaussian interference channel is exactly equal to the scaled capacity of the corresponding deterministic channel. 
 This explains why the degrees of freedom limit characterizes, up to a constant, the capacity of the Gaussian channel.


\section{Structure of Optimal Strategy for Deterministic Channel}\label{sec:StructureOfRegion}
El Gamal and Costa's characterization of the capacity region for a class of deterministic interference channels \cite{EC82} applies to this particular deterministic channel. Moreover, it is not difficult to determine the optimal input distribution from their expression. But it is not immediately apparent why this region is in fact optimal.

The goal of this section is to derive from the beginning, using only the most basic tools of information theory, the (arguably) natural optimal achievable strategy. Although the resulting strategy coincides with a specific Han and Kobayashi strategy, by proceeding in this way we hope to demystify the structure of the achievable strategy. In particular, we will see how common and private messages arise inevitably, quickly giving the capacity region of the channel. It is noteworthy that no separate outer bounds are required. Thus, the intuitive appeal of this approach is bolstered by it not requiring the side-information converse proofs of \cite{ETW07} and \cite{EC82}.

The natural decomposition of messages into common and private parts was motivated at an intuitive level for the Gaussian interference channel in Sections 6 and 7 of \cite{ETW07}. 
In the setting of the deterministic channel, the arguments of this section make those ideas precise.

The following standard definitions and notation will be used.
Denote by $\M_1=\{1,\dots,M_1\}$ and $\M_2=\{1,\dots,M_2\}$ the message sets of users 1 and 2. Let the encoding functions $f_i:\M_i\to \X_i$ with $f_i(j)=x_i(j)$ map the message $j$ generated at user $i$ into the length $N$ codeword $x_i(j)$. 
Let the decoding functions $g_i(y_i)$ map the received signal $y_i$ to the message $j$ if $y_i\in D_{ij}$, where $D_{ij}$ is the decoding set of message $j$ for user $i$. 
An $(N,M_1,M_2,\mu)$ code consists of $M_i$ codewords $x_i(j)$ and $M_i$ decoding sets $D_{ij}$ such that the average probability of decoding error satisfies
\begin{align*}
  \frac{1}{M_1 M_2}\sum_{jk} \Prob(D_{1j}|x_1(j),x_2(k))&\geq 1-\mu\,, \\
  \frac{1}{M_1 M_2}\sum_{jk} \Prob(D_{2k}|x_1(j),x_2(k))&\geq 1-\mu\,. 
\end{align*}
A pair of nonnegative real numbers  $(r_1,r_2)$ is called an \emph{achievable rate} for the deterministic interference channel if for any $\eps>0$, $0<\mu<1$, and for any sufficiently large $N$, there exists an $(N,M_1,M_2,\mu)$ code such that $$\frac{1}{N}\log M_i\geq r_i-\eps\,.$$

The first lemma is a simple analogue of Shannon's point-to-point channel coding theorem, stating that the mutual information between input and output determines the capacity region.
\begin{lemma}\label{l:BasicCodingLemma}
  The rate point $(r_1,r_2)$ is achievable if and only if for every $\eps>0$ there exists a block length $n$ and a factorized joint distribution $p(x_1^N)p(x_2^N)$ with
  \begin{equation}\begin{split}\label{e:BasicCodingTheorem}
  r_1-\eps&\leq\onen I(x_1^N;y_1^N)
  \\ r_2-\eps&\leq \onen I(x_2^N;y_2^N)\,.\end{split}\end{equation} 
\end{lemma}
\begin{proof}
  Fix a block length $N$ and joint distribution $p(x_1^N)p(x_2^N)$. Each user $i=1,2$ will use the distribution over $p(x_i^N)$ as an inner code, using $k$ blocks of length $N$. 
  The codebooks are constructed using random coding, and the achievability of $(r_1,r_2)$ follows by the random coding argument (with joint typicality decoding) for the  point-to-point discrete memoryless channel.
  
  As in the point-to-point case, the converse is a straightforward application of Fano's inequality: 
  \begin{align*}
    nr_i&=H(W_i)=H(W_i|y_i^N)+I(W_i;y_i^N)
    \\ &\leq 1+P_e^{(N)} n r_i+ I(x_i^N;y_i^N), \quad i=1,2\, .
  \end{align*}
  It is assumed that $P_e^{(N)}\to 0$ as $N\to \infty$. Dividing by $N$ and taking $N$ sufficiently large gives the desired result.
\end{proof}

\begin{figure}
\begin{centering}
\psset{unit=.7mm,arrowlength=1,arrowinset=0,labelsep=2pt}
\begin{center}
\begin{pspicture}(-1,0)(113,62)
\psline[linestyle=dashed,dash=3pt 2pt](0,0)(40,0)
\psline[linestyle=dashed,dash=3pt 2pt](70,0)(110,0)
\psline(2,0)(2,50)(14,50)(14,0) \rput(8,-5){1}
\psline(26,0)(26,30)(38,30)(38,0)\rput(32,-5){2}

\psline(72,0)(72,36)(84,36)(84,0)\rput(78,-5){1}
\psline(96,0)(96,54)(108,54)(108,0) \rput(102,-5){2}

\uput[r](14,50){ $n_{11}$} 
\uput[r](38,30){ $n_{12}$}

\uput[l](72,36){ $n_{21}$}
\uput[l](96,54){ $n_{22}$}


\psline[linestyle=dashed,dash=1pt 1pt](2,14)(14,14)
\psline[linestyle=dashed,dash=1pt 1pt](96,24)(108,24)

\rput(8,7){$x_{1p}$}
\rput(32,15){$x_{2c}$}
\rput(78,18){$x_{1c}$}
\rput(102,12){$x_{2p}$}
\rput(8,32){$x_{1c}$}
\rput(102,39){$x_{2c}$}

\rput(20,64){\large $\textbf{Rx}_1$}
\rput(90,64){\large $\textbf{Rx}_2$}

\end{pspicture}
\end{center}
\caption{The figure depicts the received signal at each receiver. Notice that the private signals (as defined in Lemma~\ref{l:separateCommonPrivate}), $x_{1p},x_{2p}$, are not observed at the other receiver.}
\label{fig:ICrectangles1}
\end{centering}
\end{figure}
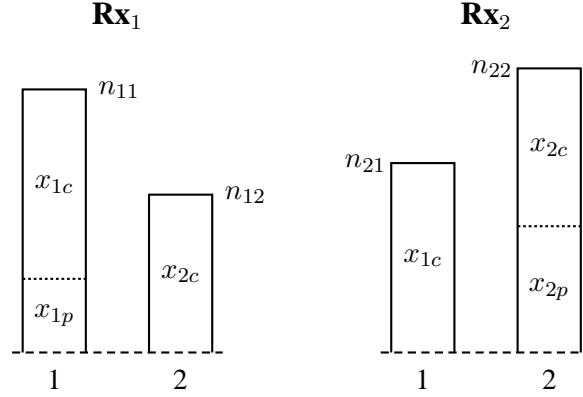

The next two lemmas are the most important of this section; they show the optimality of separating each message into a private and common message (the terms common and private are to be justified later, and for now to be regarded simply as labels).
\begin{lemma}\label{l:separateCommonPrivate}
  Given any achievable rate point $(r_1,r_2)$, this rate-point is achievable using a code with the following decomposition.
  \\1. The channel inputs, $x_1^N$ and $x_2^N$, are separated into components consisting of common and private information: $$x_1^N=(x_{1p}^N,x_{1c}^N),\quad x_2^N=(x_{2p}^N,x_{2c}^N)\, .$$
  \\2. The message sets are separated into private and common messages, i.e. $\M_i=\M_{ic}\times \M_{ip}$ for users $i=1,2$, with the common signal $x_{ic}^N=f_i^c(m_{ic})$ a function only of the common message $m_{ic}\in \M_{ic}$ and the private signal $x_{ip}^N=f_i^p(m_{ip},m_{ic})$ a function of both the private and common message $(m_{ip},m_{ic})\in \M_{ip}\times \M_{ic}$.
  \\3. The common rate is less than the entropy of the common signal, that is $r_i^c<\onen H(x_{ic}^N)$.
\end{lemma}
\begin{proof}Consider an achievable rate point $(r_1,r_2)$.
  The proof follows by converting an arbitrary achievable strategy to one that satisfies the desired properties. Fix $\eps>0$, a block length $N'$, and an arbitrary distribution $p(x_1^{N'})p(x_2^{N'})$ such that \eqref{e:BasicCodingTheorem} is satisfied with $\eps/2$. Write the input as $x_i^{N'}=(x_{ip}^{N'},x_{ic}^{N'})$, where $x_{1p}^{N'}$ is the input $x_1^{N'}$ restricted to the lowest $(n_{11}-n_{21})^+$ levels, $x_{1c}^{N'}$ is the restriction to the highest $n_{21}$ levels, and similarly for $x_{2c}^{N'},x_{2p}^{N'}$ (see Figure~\ref{fig:ICrectangles1}). Note that if $n_{21}\geq n_{11}$ ($n_{12}\geq n_{22}$) then the private signal $x_{1p}^{N'}$ ($x_{2p}^{N'}$) is empty. 
  
 It must now be verified that transmitter $i$ can separate the message set $\M_i$ into the direct product of two message sets $\M_{ip}\times \M_{ic}$. The scheme uses a superposition code, as used for the degraded broadcast channel (see e.g. \cite{CoverThomas}), with $x_{ic}$ serving as the cloud centers and $x_{ip}$ as the clouds. To see that this is possible, put for $i=1,2$, 
 \begin{equation}\label{e:privateCommonRates}
 \begin{split}
   r_i^c&=\onenp I(x_{ic}^{N'};y_i^{N'})-\frac{\eps}{4} \\
   r_i^p&=\onenp I(x_{ip}^{N'};y_i^{N'}|x_{1c}^{N'})-\frac{\eps}{4}\, . \end{split}
 \end{equation}
 Then from the chain rule we have
 \begin{align*}
   r_{ic}+r_{ip}&=\onenp I(x_{ic}^{N'};y_i^{N'})-\frac{\eps}{4} +\onenp I(x_{ip}^{N'};y_i^{N'}|x_{1c}^{N'})-\frac{\eps}{4}
   \\ &=\onenp I(x_i^{N'};y_i^{N'})-\frac{\eps}{2}\geq  r_i-\eps\, .
 \end{align*}
For some sufficiently large super-block length $k$, generate for $i=1,2$, $2^{k{N'}r_{ic}}$ independent codewords of length ${N'}k$, $x_{ic}^{k{N'}}(m_{ic})$ according to $\prod_{t=1}^k p(x_{ic,t}^{N'})$. The block-length $N$ in the statement of the lemma is given by $N=N' k$. Now, for each codeword $x_{ic}^{k{N'}}(m_{ic})$, generate $2^{k{N'}r_{ip}}$ codewords of length ${N'}k$, $x_{ip}^{k{N'}}(m_{ic},m_{ip})$, according to the conditional distribution $\prod_{t=1}^k p(x_{ip,t}^{N'}|x_{ic,t}^{N'}(m_{ic}))$. Decoding is accomplished using joint typicality, and the probability of error may be taken as small as desired by choosing $k$ large. Since $\eps$ was arbitrary, this proves the lemma.
\end{proof}

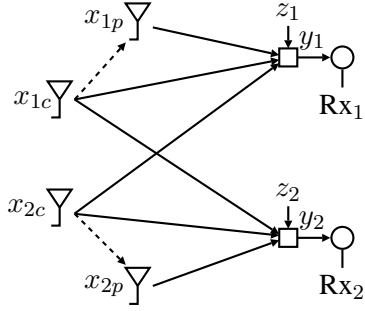
\begin{figure}
\begin{centering}
\psset{unit=.8mm,arrowlength=1,arrowinset=0}
\begin{center}
\begin{pspicture}(-16,0)(20,44)
    \rput(-16,35){%
           \rput(5,7){ \psline(0,0)(2,3)(-2,3)(0,0)
            \psline(0,0)(0,-3)(-1.4,-3)}
            
           \rput(-8,-6){ \psline(0,0)(2,3)(-2,3)(0,0)
            \psline(0,0)(0,-3)(-1.4,-3)}
            \psline{->}(7.5,6)(28.5,1.4)
            \psline{->}(-5.5,-6)(28.5,.6)
            \psline[linestyle=dashed,dash=2pt 2pt]{->}(-5.5,-6)(3,3.5)
            \pspolygon(28.5,2.5)(31.5,2.5)(31.5,-.5)(28.5,-.5)
            
            \psline{->}(31.5,1)(37,1)
            \pscircle(39,1){2}
            \psline{-}(39,-1)(39,-4)
            \uput[d](39,-3){$\text{Rx}_1$}
            \uput{3pt}[u](34,1){$y_1$}
            \psline{->}(30,6.25)(30,2.5)
            \uput[u](30,5){$z_1$}
             
            \uput[l](5,7){$x_{1p}$}
            \uput[l](-7,-6){$x_{1c}$}
            }

    \rput(-16,5){%
           \rput(-8,6){ \psline(0,0)(2,3)(-2,3)(0,0)
            \psline(0,0)(0,-3)(-1.4,-3)}
           \rput(5,-7){ \psline(0,0)(2,3)(-2,3)(0,0)
            \psline(0,0)(0,-3)(-1.4,-3)}
            \psline{->}(-5.5,5)(28.5,1.4)
            \psline[linestyle=dashed,dash=2pt 2pt]{->}(-5.5,5)(3,-3.5)
            
            \psline{->}(7.5,-7)(28.5,.6)
            \pspolygon(28.5,2.5)(31.5,2.5)(31.5,-.5)(28.5,-.5)
            
            \psline{->}(31.5,1)(37,1)
            \pscircle(39,1){2}
            \psline{-}(39,-1)(39,-4)
            \uput[l](-8,6){$x_{2c}$}
            \uput[l](5,-7){$x_{2p}$}
            \uput[d](39,-3){$\text{Rx}_2$}
            \uput{3pt}[u](34,1){$y_2$}
            \psline{->}(30,6.25)(30,2.5)
            \uput[u](30,5){$z_2$}
            }

\psline{->}(-21.5,29)(12.5,6.6)
\psline{->}(-21.5,10)(12.5,35)

\end{pspicture}
\end{center}
\caption{Lemma~\ref{l:separateCommonPrivate} shows that we may view the common signal and private signal of each user as coming from two separate users, with the private user having access to the signal from the common user.}
\label{fig:4senders}
\end{centering}
\end{figure}

The previous lemma shows that we may consider the deterministic interference channel as a channel with four senders and two decoders, as in Figure~\ref{fig:4senders}. 
This interpretation motivates the next lemma, which shows that each user is able to decode the common information of the interfering user.  The lemma makes use of facts concerning the multiple access channel. For background on the multiple access channel see e.g. \cite{CsiszarKornerBook, CoverThomas}. The lemma can essentially be deduced from the result by Costa and El Gamal on discrete memoryless interference channels with strong interference \cite{CE87}. The result itself is analogous to Sato's result for the Gaussian interference channel in the strong interference regime \cite{S81}; however, because Lemma~\ref{l:separateCommonPrivate} shows that the signal ought to be separated into common and private components, the argument applies to the entire parameter range.

By the MAC at receiver $1$ we mean the MAC formed by the two users transmitting signals $(x_{1p},x_{1c})$ and $x_{2c}$ at rates $r_1^p+r_1^c$ and $r_2^c$, respectively, with receiver $1$ required to reliably decode both signals, and similarly for the MAC at receiver $2$.

\begin{lemma}\label{l:MACintersection}
  The region is exactly described by the compound MAC formed by the MAC at each of the two receivers, along with constraints on the private rate.
  Furthermore, the region has a single-letter representation. 
\end{lemma}
\begin{proof}
  Suppose the rate-point $(r_1,r_2)$ is achievable. By Lemma~\ref{l:separateCommonPrivate}, we may assume that each user's common signal is a function only of the common message, and that \begin{equation}\label{e:commonRateSmall}
  \begin{split}
  r_1^c&=\onen I(x_{1c}^N;y_1^N)-\eps\leq \frac{1}{N} H(x_{1c}^N)-\eps \\
  r_2^c&=\onen I(x_{2c}^N;y_2^N)-\eps\leq \frac{1}{N} H(x_{2c}^N)-\eps \,.\end{split}\end{equation}
  Then each user, upon successfully decoding their own signal and subtracting it off, has a clear view of the other user's common signal $x_{ic}^N$. But, since the common rate is smaller than the entropy of the common signal \eqref{e:commonRateSmall}, it is possible to recover the common message $m_{ic}$ with arbitrarily small probability of error when $N$ is taken to be large enough; in other words, each user can reliably decode the other user's common message. 
  
  The joint distribution of the channel is \begin{equation}\begin{split}\label{e:jointdist}
  &p(y_1^N|x_{1c}^N,x_{1p}^N,x_{2c}^N)p(y_2^N|x_{2c}^N,x_{2p}^N,x_{1c}^N)\\ &\qquad p(x_{1p}^N|x_{1c}^N)p(x_{1c}^N) p(x_{2p}^N|x_{2c}^N)
  p(x_{2c}^N)\,.\end{split}\end{equation} The fact that each receiver can decode the common message of the other user implies, by Fano's inequality, that 
  $$\onen H(m_{1c},m_{1p},m_{2c}|y_1^N)\to 0$$ and
  $$\onen H(m_{1c},m_{2p},m_{2c}|y_2^N)\to 0$$
  as $N\to \infty$. 
  
  Proceeding as in the converse argument for the MAC (see e.g. \cite{CoverThomas}), one can show that for any joint distribution \eqref{e:jointdist} the rate point $(r_1^c,r_1^p,r_2^c,r_2^p)$ satisfies a number of constraints. First, the rate point $(r_1^c+r_1^p,r_2^c)$ must lie within the MAC at receiver 1 and the rate point $(r_1^c,r_2^c+r_2^p)$ must lie within the MAC at receiver 2. Additionally, there are constraints on the private rates $r_1^p,r_2^p$ and the rates $r_1^p+r_2^c$ and $r_2^p+r_1^c$.  More precisely, there exists a distribution $p(x_{1p}|x_{1c},q)p(x_{1c}|q) p(x_{2p}|x_{2c},q)
  p(x_{2c}|q)p(q)$ such that
  \begin{equation}\begin{split}\label{e:MACwithTimeSharing}
  r_1+r_2^c=r_1^c+r_1^p+r_2^c &\leq I(x_{1c},x_{1p},x_{2c};y_1|Q)
  \\
  r_1=r_1^c +r_1^p &\leq I(x_{1c},x_{1p};y_1|x_{2c},Q)
  \\
  r_2^c &\leq I(x_{2c};y_1|x_{1c},x_{1p},Q)
  \\
  r_1^p+r_2^c &\leq I(x_{1p},x_{2c};y_1|x_{1c},Q)
  \\
  r_1^p&\leq I(x_{1p};y_1|x_{2c},x_{1c},Q)
  \\
  r_1^c+r_2=r_1^c+r_2^p+r_2^c &\leq I(x_{2c},x_{2p},x_{1c};y_2|Q)
  \\
  r_2=r_2^c +r_2^p &\leq I(x_{2p},x_{2c};y_2|x_{1c},Q)
  \\
  r_1^c &\leq I(x_{1c};y_2|x_{2c},x_{2p},Q)
  \\
  r_2^p+r_1^c &\leq I(x_{2p},x_{1c};y_2|x_{2c},Q)
  \\
  r_2^p&\leq I(x_{2p};y_2|x_{1c},x_{2c},Q)\, .
  \end{split}\end{equation}

Conversely, if the rate tuple $(r_1^c+r_1^p,r_2^c)$ is within the MAC at receiver 1, and $(r_1^c,r_2^c+r_2^p)$ is within the MAC at receiver 2, and the additional constraints on $r_1^p,r_2^p$ are satisfied, then the rate point $(r_1^c,r_1^p,r_2^c,r_2^p)$ is achievable using a superposition random code as in Lemma~\ref{l:separateCommonPrivate} and joint typicality decoding.
  \end{proof}
 The next lemma makes the region in equation \eqref{e:MACwithTimeSharing} explicit.
  
  \begin{lemma}\label{l:explicitMACregion}
  Observe that the optimizing (simultaneously for each of the constraints in \eqref{e:MACwithTimeSharing}) input distribution is uniform for each signal. This allows us to write the region as
  \begin{equation}\begin{split}\label{e:explicitMACregion}
  r_1^c+r_1^p+r_2^c &\leq n_{11} + \min(n_{22},(n_{12}-n_{11})^+)
  \\
  r_1^c +r_1^p &\leq n_{11}
  \\
  r_2^c &\leq  \min(n_{12},n_{22}) 
  \\
  r_1^p+r_2^c &\leq \min(n_{22}+(n_{11}-n_{21})^+, 
    n_{12})
  \\
  r_1^p &\leq  n_{11}-n_{21} 
  \\
  r_2^c+r_2^p+r_1^c &\leq 
    n_{22}+\min(n_{11},(n_{21}-n_{22})^+)
  \\
  r_2^c +r_2^p &\leq n_{22}
  \\
  r_1^c &\leq  \min(n_{21},n_{11})
  \\
  r_2^p+r_1^c &\leq  \min(n_{11}+(n_{22}-n_{12})^+, 
    n_{21})
  \\
  r_2^p &\leq  n_{22}-n_{12} 
\, .
  \end{split}\end{equation}  
  \end{lemma}
  \begin{proof}
  Intuitively, the private signal should be uniform because it helps the intended receiver decode and does not cause interference, and the common signal should be uniform because it helps both receivers decode.
  
  Fix a joint distribution and consider a rate point satisfying the constraints of the previous lemma. 
  From the equations of the previous lemma, it is easy to see that $p(x_{ip})$ should be uniform in any optimal distribution, since this increases the mutual information terms where $x_{ip}$ appears. Similarly, $p(x_{ic})$ should be uniform. This allows to evaluate the mutual information expressions in equation \eqref{e:MACwithTimeSharing}, resulting in the stated region.
  \end{proof}
\begin{remark}
  The constraints of Lemma~\ref{l:explicitMACregion} admit a simple interpretation in terms of the areas of the relevant rectangles in Figure~\ref{fig:ICrectangles3}. 
\end{remark}

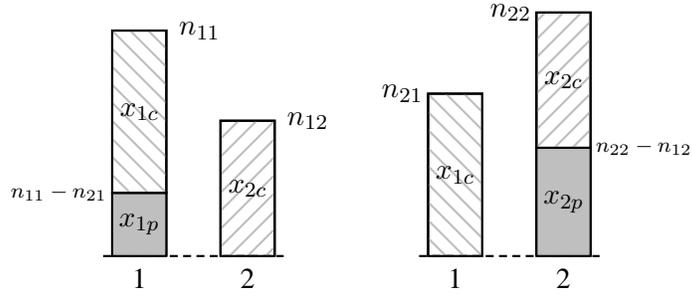
\begin{figure}
\begin{centering}
\psset{unit=.6mm,arrowlength=1,arrowinset=0,labelsep=2pt}
\begin{center}
\begin{pspicture}(-1,0)(113,62)
\psline[linestyle=dashed,dash=3pt 2pt](0,0)(40,0)
\psline[linestyle=dashed,dash=3pt 2pt](70,0)(110,0)
\psline(2,0)(2,50)(14,50)(14,0) \rput(8,-5){1}
\psline(26,0)(26,30)(38,30)(38,0)\rput(32,-5){2}

\psline(72,0)(72,36)(84,36)(84,0)\rput(78,-5){1}
\psline(96,0)(96,54)(108,54)(108,0) \rput(102,-5){2}

\uput[r](14,50){ $n_{11}$} 
\uput[r](38,30){ $n_{12}$}

\uput[l](72,36){ $n_{21}$}
\uput[l](96,54){ $n_{22}$}



\pspolygon[fillstyle=solid, fillcolor=lightgray](2,0)(2,14)(14,14)(14,0)

\pspolygon[fillstyle=solid, fillcolor=lightgray](96,0)(96,24)(108,24)(108,0)

\pspolygon[fillstyle=vlines,
hatchcolor=lightgray](2,14)(2,50)(14,50)(14,14)

\pspolygon[fillstyle=hlines, hatchcolor=lightgray](96,24)(96,54)(108,54)(108,24)

\pspolygon[fillstyle=vlines,
hatchcolor=lightgray](72,0)(72,36)(84,36)(84,0)

\pspolygon[fillstyle=hlines,
hatchcolor=lightgray](26,0)(26,30)(38,30)(38,0)

\rput(8,7){$x_{1p}$}
\rput(32,15){$x_{2c}$}
\rput(78,18){$x_{1c}$}
\rput(102,12){$x_{2p}$}
\rput(8,32){$x_{1c}$}
\rput(102,39){$x_{2c}$}


\uput[l](2,14){\scriptsize $n_{11}-n_{21}$}
\uput[4](108,24){\scriptsize $n_{22}-n_{12}$}



\end{pspicture}
\end{center}
\caption{From the figure it is possible to understand the constraints \eqref{e:explicitMACregion} as areas of rectangles.}
\label{fig:ICrectangles3}
\end{centering}
\end{figure}

The constraints \eqref{e:explicitMACregion} determine the capacity region of the deterministic channel; using Fourier-Motzkin elimination one can solve for the region in terms of constraints on $r_1$ and $r_2$. Alternatively, note that the deterministic interference channel of this paper falls within the class of more general deterministic channels whose capacity is given in Theorem 1 of~\cite{EC82}. Applying this theorem, the deterministic channel capacity region is the set of nonnegative rates satisfying
\begin{align*}
  r_i&\leq n_{ii},\quad i=1,2 \\
  r_1+r_2&\leq (n_{11}-n_{12})^++\max(n_{22},n_{12})\\
  r_1+r_2&\leq (n_{22}-n_{21})^+ +\max(n_{11},n_{21})\\
  r_1+r_2&\leq \max(n_{21},(n_{11}-n_{12})^+) +\max(n_{12},(n_{22}-n_{21})^+)\\
  2r_1+r_2&\leq \max(n_{11},n_{21})+(n_{11}-n_{12})^+ +\max(n_{12},(n_{22}-n_{21})^+) \\
  r_1+2r_2&\leq \max(n_{22},n_{12})+(n_{22}-n_{21})^+ +\max(n_{21},(n_{11}-n_{12})^+)\,.
\end{align*}



\section{Examples}\label{sec:AchievableSchemeExamples}
It is instructive to consider a few examples of capacity-achieving schemes for the deterministic channel. For simplicity, we restrict attention to the symmetric case, i.e. $n:=n_{11}=n_{22}$ and $n_{21}=n_{12}=n\al$, where $\al:=\frac{n_{12}}{n_{11}}$. Most of the achievable schemes presented admit simple interpretations in the Gaussian channel. Figure~\ref{fig:Wcurve} depicts the sum-rate capacity of the symmetric channel, indexed by $\alpha$. 
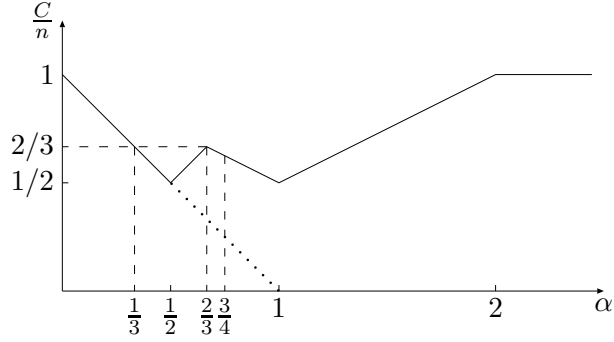
\begin{figure}
\begin{centering}
\psset{unit=.8mm,linewidth=.3pt,arrowlength=1.2,
arrowinset=0,labelsep=3pt}
\begin{center}
\begin{pspicture}(0,0)(90,50)

\psline{->}(0,0)(90,0)
\psline{->}(0,0)(0,45)

\uput[d](90,0){$\alpha$}
\uput[l](0,45){$\frac{C}{n}$}

\psline(0,36)(18,18)(24,24)(36,18)(72,36)(88,36)

\psline[linestyle=dashed, dash= 4pt 1pt](12,0)(12,24)
\psline[linestyle=dashed, dash= 4pt 1pt](24,0)(24,24)
\psline[linestyle=dashed, dash= 4pt 1pt](0,24)(24,24)
\psline[linestyle=dashed, dash= 4pt 1pt](27,0)(27,22.5)
\psline[linestyle=dotted,linewidth=1pt](18,18)(36,0)

\uput[l](0,18){$1/2$}
\uput[l](0,24){$2/3$}
\uput[l](0,36){$1$}
\uput[d](12,0){$\frac{1}{3}$}
\uput[d](18,0){$\frac{1}{2}$}
\uput[d](24,0){$\frac{2}{3}$}
\uput[d](27,0){$\frac{3}{4}$}
\uput[d](36,0){$1$}
\uput[d](72,0){$2$}

\psline(18,0)(18,1)
\psline(36,0)(36,1)
\psline(72,0)(72,1)
\psline(0,18)(1,18)

\end{pspicture}
\end{center} 
\caption{The sum-rate capacity of the deterministic interference channel, normalized by $n$. The dotted line continuing downwards from the point $(1/2,1/2)$ is the rate achievable by treating interference as noise.}
\label{fig:Wcurve}
\end{centering}
\end{figure}

Consider first the case $\al=1/3$. One option is to use the strategy described in Section~\ref{sec:StructureOfRegion}, making the entire signal private information (Figure~\ref{fig:exampleAcheivable1}). In the deterministic model the signal does not appear at the unintended receiver. This corresponds to transmitting below the noise level in the Gaussian channel, in which case the additional noise from the interference causes a loss of only one bit for each user. A second option is for each transmitter to use the full available power, transmitting on the highest $2/3$ of the levels (Figure~\ref{fig:exampleAchievable1b}). The lower $1/3$ of the levels are unusable on the direct link due to the presence of interference. This strategy corresponds to treating interference as noise in the Gaussian channel. The value $\al=1/3$ is representative of the entire range $\al\in [0,\frac{1}{2}]$, where both of these strategies are optimal.

\begin{figure}
\begin{centering}
\psset{unit=.6mm,arrowlength=1,arrowinset=0,labelsep=2pt}
\begin{center}
\begin{pspicture}(-1,4)(113,50)
\psline[linestyle=dashed,dash=3pt 2pt](0,0)(40,0)
\psline[linestyle=dashed,dash=3pt 2pt](70,0)(110,0)
\psline(2,0)(2,30)(14,30)(14,0) \rput(8,-5){1}
\psline(26,0)(26,10)(38,10)(38,0)\rput(32,-5){2}

\psline(72,0)(72,10)(84,10)(84,0)\rput(78,-5){1}
\psline(96,0)(96,30)(108,30)(108,0) \rput(102,-5){2}

\uput[r](14,30){ $n$} 
\uput[r](38,10){ $\frac{n}{3}$}

\uput[l](72,10){ $\frac{n}{3}$}
\uput[l](96,30){ $n$}

\pspolygon[fillstyle=solid, fillcolor=lightgray](2,0)(2,20)(14,20)(14,0)

\pspolygon[fillstyle=solid, fillcolor=lightgray](96,0)(96,20)(108,20)(108,0)

\rput(8,10){$x_{1p}$}
\rput(102,10){$x_{2p}$}

\rput(20,45){\large $\textbf{Rx}_1$}
\rput(90,45){\large $\textbf{Rx}_2$}

\end{pspicture}
\end{center}
\caption{$\alpha=1/3$. Two-thirds of the signal is private information, with no common information. This scheme corresponds to transmitting below the noise level.}
\label{fig:exampleAcheivable1}
\psset{unit=.6mm,arrowlength=1,arrowinset=0,labelsep=2pt}
\begin{center}
\begin{pspicture}(-1,4)(113,50)
\psline[linestyle=dashed,dash=3pt 2pt](0,0)(40,0)
\psline[linestyle=dashed,dash=3pt 2pt](70,0)(110,0)

\psline(2,0)(2,30)(14,30)(14,0) \rput(8,-5){1}
\psline(26,0)(26,10)(38,10)(38,0)\rput(32,-5){2}

\psline(72,0)(72,10)(84,10)(84,0)\rput(78,-5){1}
\psline(96,0)(96,30)(108,30)(108,0) \rput(102,-5){2}

\uput[r](14,30){ $n$} 
\uput[r](38,10){ $\frac{n}{3}$}

\uput[l](72,10){ $\frac{n}{3}$}
\uput[l](96,30){ $n$}

\pspolygon[fillstyle=solid, fillcolor=lightgray](26,10)(26,0)(38,0)(38,10)
\pspolygon[fillstyle=solid, fillcolor=lightgray](2,10)(2,20)(14,20)(14,10)
\pspolygon[fillstyle=solid, fillcolor=lightgray](2,30)(2,20)(14,20)(14,30)

\pspolygon[fillstyle=solid, fillcolor=lightgray](72,10)(72,0)(84,0)(84,10)
\pspolygon[fillstyle=solid, fillcolor=lightgray](96,10)(96,20)(108,20)(108,10)
\pspolygon[fillstyle=solid, fillcolor=lightgray](96,30)(96,20)(108,20)(108,30)

\rput(8,15){$x_{1p}$}
\rput(8,25){$x_{1c}$}
\rput(32,5){$x_{2c}$}
\rput(78,5){$x_{1c}$}
\rput(102,15){$x_{2p}$}
\rput(102,25){$x_{2c}$}

\rput(20,45){\large $\textbf{Rx}_1$}
\rput(90,45){\large $\textbf{Rx}_2$}

\end{pspicture}
\end{center}
\caption{$\alpha=1/3$. The top third of the levels are common information, and the middle third are private information. This scheme corresponds to treating interference as noise.}
\label{fig:exampleAchievable1b}
\end{centering}
\end{figure}

For $\al=2/3$ there are again a few options. One possibility is to use the capacity achieving scheme of Section~\ref{sec:StructureOfRegion}, with the lowest $1/3$ of the levels consisting of private information, and the remaining $2/3$ of the levels as common information (see Figure~\ref{fig:exampleAcheivable2}). The rate achieved is $r_1=r_2=2 n/3$ bits per channel use per user. Alternatively, imagine continuously varying $\alpha$ from the value $\al=1/3$ to $\al=2/3$, while using the scheme of treating interference as noise (Figure~\ref{fig:exampleAchievable1b}). The used power range will shrink to the range between $2n/3$ and $n$. However, a gap appears, and the range of levels between $1$ and $n/3$ can be used as well (Figure~\ref{fig:exampleAcheivable2b}). The gap in the corresponding Gaussian setting is because of the structure of the interference: the interference contains information, and can be decoded. After decoding the interference it can be subtracted off, and additional information can be transmitted. This phenomenon is the reason why treating interference as noise is no longer optimal beyond $\al=1/2$. 
\begin{figure}
\begin{centering}
\psset{unit=.6mm,arrowlength=1,arrowinset=0,labelsep=2pt}
\begin{center}
\begin{pspicture}(-1,4)(113,50)
\psline[linestyle=dashed,dash=3pt 2pt](0,0)(40,0)
\psline[linestyle=dashed,dash=3pt 2pt](70,0)(110,0)

\psline(2,0)(2,30)(14,30)(14,0) \rput(8,-5){1}
\psline(26,0)(26,20)(38,20)(38,0)\rput(32,-5){2}

\psline(72,0)(72,20)(84,20)(84,0)\rput(78,-5){1}
\psline(96,0)(96,30)(108,30)(108,0) \rput(102,-5){2}

\uput[r](14,30){ $n$} 
\uput[r](38,20){ $\frac{2n}{3}$}

\uput[l](72,20){ $\frac{2n}{3}$}
\uput[l](96,30){ $n$}

\pspolygon[fillstyle=solid, fillcolor=lightgray](2,0)(2,10)(14,10)(14,0)

\pspolygon[fillstyle=solid, fillcolor=lightgray](96,0)(96,10)(108,10)(108,0)

\pspolygon[fillstyle=vlines,
hatchcolor=lightgray](2,10)(2,30)(14,30)(14,10)

\pspolygon[fillstyle=hlines, hatchcolor=lightgray](96,10)(96,30)(108,30)(108,10)

\pspolygon[fillstyle=vlines,
hatchcolor=lightgray](72,0)(72,20)(84,20)(84,0)

\pspolygon[fillstyle=hlines,
hatchcolor=lightgray](26,0)(26,20)(38,20)(38,0)

\rput(8,5){$x_{1p}$}
\rput(32,10){$x_{2c}$}
\rput(78,10){$x_{1c}$}
\rput(102,5){$x_{2p}$}
\rput(8,20){$x_{1c}$}
\rput(102,20){$x_{2c}$}

\rput(20,45){\large $\textbf{Rx}_1$}
\rput(90,45){\large $\textbf{Rx}_2$}

\end{pspicture}
\end{center}
\caption{$\alpha=2/3$. One-third of the signal is private information, and two-thirds is common information, but the common rate equals the private rate: $r_1^p=r_2^p=r_1^c=r_2^c=n/3$.}
\label{fig:exampleAcheivable2}
%
\psset{unit=.6mm,arrowlength=1,arrowinset=0,labelsep=2pt}
\begin{center}
\begin{pspicture}(-1,4)(113,50)
\psline[linestyle=dashed,dash=3pt 2pt](0,0)(40,0)
\psline[linestyle=dashed,dash=3pt 2pt](70,0)(110,0)

\psline(2,0)(2,30)(14,30)(14,0) \rput(8,-5){1}
\psline(26,0)(26,20)(38,20)(38,0)\rput(32,-5){2}

\psline(72,0)(72,20)(84,20)(84,0)\rput(78,-5){1}
\psline(96,0)(96,30)(108,30)(108,0) \rput(102,-5){2}

\uput[r](14,30){ $n$} 
\uput[r](38,20){ $\frac{2n}{3}$}

\uput[l](72,20){ $\frac{2n}{3}$}
\uput[l](96,30){ $n$}

\pspolygon[fillstyle=solid, fillcolor=lightgray](26,10)(26,20)(38,20)(38,10)
\pspolygon[fillstyle=solid, fillcolor=lightgray](2,10)(2,0)(14,0)(14,10)
\pspolygon[fillstyle=solid, fillcolor=lightgray](2,30)(2,20)(14,20)(14,30)

\pspolygon[fillstyle=solid, fillcolor=lightgray](72,10)(72,20)(84,20)(84,10)
\pspolygon[fillstyle=solid, fillcolor=lightgray](96,10)(96,0)(108,0)(108,10)
\pspolygon[fillstyle=solid, fillcolor=lightgray](96,30)(96,20)(108,20)(108,30)

\rput(8,5){$x_{1p}$}
\rput(8,25){$x_{1c}$}
\rput(32,15){$x_{2c}$}
\rput(78,15){$x_{1c}$}
\rput(102,5){$x_{2p}$}
\rput(102,25){$x_{2c}$}

\rput(20,45){\large $\textbf{Rx}_1$}
\rput(90,45){\large $\textbf{Rx}_2$}

\end{pspicture}
\end{center}
\caption{$\alpha=2/3$. As $\al$ is increased from $1/3$ to $2/3$, a gap appears in the bottom 1/3 of the levels. This gap can be used to transmit private information.}
\label{fig:exampleAcheivable2b}
\end{centering}

\begin{centering}
\psset{unit=.6mm,arrowlength=1,arrowinset=0,labelsep=2pt}
\begin{center}
\begin{pspicture}(-1,4)(113,60)
\psline[linestyle=dashed,dash=3pt 2pt](0,0)(40,0)
\psline[linestyle=dashed,dash=3pt 2pt](70,0)(110,0)

\psline(2,0)(2,40)(14,40)(14,0) \rput(8,-5){1}
\psline(26,0)(26,30)(38,30)(38,0)\rput(32,-5){2}

\psline(72,0)(72,30)(84,30)(84,0)\rput(78,-5){1}
\psline(96,0)(96,40)(108,40)(108,0) \rput(102,-5){2}

\uput[r](14,40){ $n$} 
\uput[r](38,30){ $\frac{3n}{4}$}

\uput[l](72,30){ $\frac{3n}{4}$}
\uput[r](108,40){ $n$}

\pspolygon[fillstyle=solid, fillcolor=lightgray](2,0)(2,10)(14,10)(14,0)

\pspolygon[fillstyle=solid, fillcolor=lightgray](96,0)(96,10)(108,10)(108,0)
\pspolygon[fillstyle=vlines,
hatchcolor=lightgray](2,10)(2,40)(14,40)(14,10)

\pspolygon[fillstyle=hlines, hatchcolor=lightgray](96,10)(96,40)(108,40)(108,10)

\pspolygon[fillstyle=vlines,
hatchcolor=lightgray](72,0)(72,30)(84,30)(84,0)

\pspolygon[fillstyle=hlines,
hatchcolor=lightgray](26,0)(26,30)(38,30)(38,0)

\rput(8,5){$x_{1p}$}
\rput(32,15){$x_{2c}$}
\rput(78,15){$x_{1c}$}
\rput(102,5){$x_{2p}$}
\rput(8,25){$x_{1c}$}
\rput(102,25){$x_{2c}$}

\rput(20,55){\large $\textbf{Rx}_1$}
\rput(90,55){\large $\textbf{Rx}_2$}

\end{pspicture}
\end{center}
\caption{$\alpha=3/4$. This scheme is essentially the same as in Figure~\ref{fig:exampleAcheivable2}. One-quarter of the signal is private information and three-quarters is common information. The common rate is $r_1^c=r_2^c=3n/8$ and the private rate is $r_1^p=r_2^p=n/4$.}
\label{fig:exampleAcheivable3}
%
%
\psset{unit=.6mm,arrowlength=1,arrowinset=0,labelsep=2pt}
\begin{center}
\begin{pspicture}(-1,4)(113,60)
\psline[linestyle=dashed,dash=3pt 2pt](0,0)(40,0)
\psline[linestyle=dashed,dash=3pt 2pt](70,0)(110,0)

\psline(2,0)(2,40)(14,40)(14,0) \rput(8,-5){1}
\psline(26,0)(26,30)(38,30)(38,0)\rput(32,-5){2}

\psline(72,0)(72,30)(84,30)(84,0)\rput(78,-5){1}
\psline(96,0)(96,40)(108,40)(108,0) \rput(102,-5){2}

\uput[r](14,40){ $n$} 
\uput[r](38,30){ $\frac{3n}{4}$}

\uput[l](72,30){ $\frac{3n}{4}$}
\uput[r](108,40){ $n$}


\psline[linestyle=dashed,dash=1pt 1pt](2,10)(14,10)
\psline[linestyle=dashed,dash=1pt 1pt](2,20)(14,20)
\psline[linestyle=dashed,dash=1pt 1pt](2,30)(14,30)

\psline[linestyle=dashed,dash=1pt 1pt](26,10)(38,10)
\psline[linestyle=dashed,dash=1pt 1pt](26,20)(38,20)

\psline[linestyle=dashed,dash=1pt 1pt](72,10)(84,10)
\psline[linestyle=dashed,dash=1pt 1pt](72,20)(84,20)

\psline[linestyle=dashed,dash=1pt 1pt](96,10)(108,10)
\psline[linestyle=dashed,dash=1pt 1pt](96,20)(108,20)
\psline[linestyle=dashed,dash=1pt 1pt](96,30)(108,30)

\rput(8,5){$a_1$}
\rput(8,15){$b_1$}
\rput(8,25){$b_1$}
\rput(8,35){$c_1$}

\rput(32,25){$b_2$}

\rput(78,5){$b_1$}
\rput(78,15){$b_1$}
\rput(78,25){$c_1$}

\rput(102,5){$a_2$}
\rput(102,35){$b_2$}

\rput(20,55){\large $\textbf{Rx}_1$}
\rput(90,55){\large $\textbf{Rx}_2$}

\end{pspicture}
\end{center}
\caption{$\alpha=3/4$. Coding over levels is performed by repeating the vector of bits $b_1$.}
\label{fig:exampleAcheivable3b}
\end{centering}
\end{figure}

The case $\al=3L/4$ is different than the previous examples: here coding is necessary. The random code of Section~\ref{sec:StructureOfRegion} has the lowest $1/4$ of the levels containing private information and the highest $3/4$ of the levels contain common information (Figure~\ref{fig:exampleAcheivable3}). The symmetric rate achieved is $5n/8$ bits per channel use per user. 
As in the previous examples, using only one time-slot is possible, but for $\al>2/3$, using one time-slot requires coding over \emph{levels}. The scheme in \cite{BT08}, shown in Figure~\ref{fig:exampleAcheivable3b}, achieves the rate point $(3n/4,n/2)$ by repeating a symbol on two different levels; the symmetric point $(5n/8,5n/8)$ is achieved by time-sharing.


\appendix
\section{Proof of Deterministic Approximation Theorem}\label{sec:DeterministicApproxProof}
In this appendix we prove Theorem~\ref{t:DeterministicApproximation}, which states that the capacity region of the 2-user Gaussian interference channel is within $42$ bits per user of the deterministic interference channel. More specifically, for each choice of channel parameters in the Gaussian channel, the corresponding deterministic channel has approximately the same capacity region. The focus is not on optimizing the size of the gap; several of the estimates are weakened in favor of a simpler argument. Rather, the significance is that the gap is \emph{constant}, independent of the channel gains. Moreover, the proof uses no knowledge of the Gaussian channel. Thus, the approach used here, along with the deterministic capacity region from  Section~\ref{sec:StructureOfRegion}, gives an alternative derivation of the constant gap capacity result of Etkin, Tse, and Wang \cite{ETW07}.

We first prove Theorem~\ref{t:Appendix:DetApproxReal}, which is the same as Theorem~\ref{t:DeterministicApproximation} but for the \emph{real} Gaussian interference channel, where the inputs, channel gains, and noise are real-valued. The complex-valued case is discussed afterwards. The main ingredients used in the proof of Theorem~\ref{t:DeterministicApproximation} for the complex-valued channel are the same as those introduced in the proof of the real-valued channel.

\begin{theorem}
\label{t:Appendix:DetApproxReal}
  The capacity of the \emph{real-valued} 2-user Gaussian interference channel with signal and interference
to noise ratios $\snr_1,\snr_2,\inr_1,\inr_2$ is within $18.6$ bits per
user of the capacity of a deterministic interference channel with
gains $2^{n_{11}}:=2^{ \lfloor \half\log \snr_1\rfloor }$, $2^{n_{12}}:=2^{ \lfloor \half\log \inr_2 \rfloor}$, $2^{n_{21}}:=2^{ \lfloor \half\log \inr_1\rfloor}$, and $2^{n_{22}}:=2^{ \lfloor \half \log  \snr_2\rfloor}$.
\end{theorem}
The factor of $\half$ in front of the logarithm is due to the channel being real-valued. 

Recall that the real-valued Gaussian interference channel is given by 
\begin{equation}\label{e:GaussianChannelAppendix}\begin{split}
  y_1&=h_{11}x_1+h_{12}x_2+z_1
  \\
  y_2&=h_{21}x_1+h_{22}x_2+z_2\end{split}
\end{equation}
where $z_i\sim \N(0,1)$, $h_{ij}\in \R$, and the input signals $x_1,x_2$ satisfy an average power constraint $$\onen\sum_{k=1}^n\E[x_{i,k}^2]\leq P_i\,.$$
By scaling the channel gains, we may assume without loss of generality that the average power constraints of the Gaussian channel are equal to 1, i.e. $P_1=P_2=1$. 

The corresponding deterministic channel, introduced in Section~\ref{sec:deterministicIC}, is 
\begin{equation}\begin{split}\label{e:Appendix:detChannel}
  y_1&= \lfloor 2^{n_{11}} x_1 \rf \op \lf 2^{n_{12}}x_2 \rfloor \\
  y_2&= \lfloor 2^{n_{21}} x_1 \rf \op \lf 2^{n_{22}}x_2 \rfloor\, ,
\end{split}\end{equation}
where $n_{ij}=\lf \log |h_{ij}| \rf$ and $x_i,i=1,2$ are real numbers, $0\leq x_i\leq 1$. Addition is modulo 2 in each position in the binary expansion.

The proof of Theorem~\ref{t:Appendix:DetApproxReal} requires two directions, namely $$C_{Gaussian}\subseteq C_{det}+\text{constant}$$
and
$$C_{det}\subseteq C_{Gaussian}+\text{constant}\, .$$
Each direction will be completed in a sequence of steps, each step comparing the capacity region of a new channel to that of the previous step. The first and last channels will be the Gaussian and deterministic channels under our consideration.



\subsection{$C_{det}\subseteq C_{Gaussian}+(5,5)$}
\label{subsec:Appendix:DetToGaussian}
We now show that the capacity achieving input of the deterministic channel \eqref{e:Appendix:detChannel} can be transferred over to the Gaussian channel \eqref{e:GaussianChannelAppendix} with a loss of at most 5 bits per user. This specifies an achievable region for the Gaussian channel. As mentioned above, the argument is based on comparing mutual information in a sequence of steps. 

The first step shows that the capacity region does not decrease if the modulo 2 addition of the deterministic channel is replaced by real addition; Step 2 shows that the capacity region of the deterministic channel~\eqref{e:Appendix:detChannel} is the same if the gain $2^{n_{ij}}$ is replaced by a real-valued $h_{ij}$ with $n_{ij}=\lf \log |h_{ij}|\rf$; Step 3 adds Gaussian noise; Step 4 removes the truncation of received signals at the noise level.

The following easy lemma bounds the effect of a change to the channel output when the original output can be restored using a small amount of side information, and will be used several times.
\begin{lemma}
  \label{l:Appendix:sideInformationLemma}
  Fix a block-length $N$. If the signal $y^N$ is determined by the pair $\y^N,s^N$, then $$I(x^N;\y^N)\geq I(x^N;y^N)-H(s^N)\,.$$
\end{lemma}
\begin{proof}
The assumption that $y^N$ is determined by $\y^N,s^N$ implies that $$H(x^N|\y^N,s^N)\leq H(x^N|y^N)\,.$$ 
This inequality together with the chain rule gives
   \begin{align*} I(x^N;\y^N)&=H(x^N)-H(x^N|\y^N)\\
   &\geq H(x^N)-H(x^N,s^N|\y^N)\\
   &\geq H(x^N)-H(s^N)-H(x^N|\y^N,s^N)\\
   &\geq H(x^N)-H(s^N)-H(x^N|y^N)\\
   &=I(x^N;y^N)-H(s^N) \,.
   \end{align*} This proves the lemma.
\end{proof}


\vspace{2mm}
\noindent 
{\bf Step 1: Real addition (lose zero bits).} For simplicity, only the output $y_1$ is discussed. The corresponding statements for $y_2$ follow similarly. 

We may write the inputs as
\begin{equation}\label{e:Appendix:InputBinary}
  x_i=\sum_{k=1}^\infty x_i(k) 2^{-k},\quad x_i(k)\in \{0,1\}\,.
\end{equation}
 In the deterministic channel \eqref{e:Appendix:detChannel}, 
 we have 
 \begin{align*}\lf 2^{n_{11}} \sum_{k=1}^\infty x_1(k) 2^{-k}\rf = \sum_{k=1}^{n_{11}} 2^{n_{11}-k}x_1(k)
 \end{align*}
 and
  \begin{align*}\lf 2^{n_{12}} \sum_{k=1}^\infty x_2(k) 2^{-k}\rf
= \sum_{k=1}^{n_{12}} 2^{n_{12}-k}x_2(k)\,.
 \end{align*}
 Thus, the common signal from user 2 is $$x_{2c}=\{x_2(1),\dots,x_2(n_{12})\}\,.$$

Step 1 replaces the modulo 2 addition of the deterministic channel with real addition. 
Using the two previous equations, we define (the output at receiver 1 of) Channel 1 as
\begin{equation*}
  y_1=\sum_{k=1}^{n_{11}} 2^{n_{11}-k}x_1(k) +\sum_{k=1}^{n_{12}} 2^{n_{12}-k}x_2(k)\,.
\end{equation*}

We claim that the capacity region of this new channel contains that of the original deterministic channel. Any rate point within the region \eqref{e:MACwithTimeSharing} given by Lemma~\ref{l:MACintersection} is achievable for Channel~1:
 \begin{equation}\begin{split}\label{e:Appendix:MAC}
  r_1^c+r_1^p+r_2^c &\leq I(x_{1c},x_{1p},x_{2c};y_1)=H(y_1)
  \\
  r_1^c +r_1^p &\leq I(x_{1c},x_{1p};y_1|x_{2c})=H(x_1)
  \\
  r_2^c &\leq I(x_{2c};y_1|x_{1c},x_{1p})=H(x_{2c})
  \\
  r_1^p+r_2^c &\leq I(x_{1p},x_{2c};y_1|x_{1c})=H(y_1|x_{1c})
  \\
  r_1^p&\leq I(x_{1p};y_1|x_{2c},x_{1c})=H(x_{1p})
  \\
  r_1^c+r_2^p+r_2^c &\leq I(x_{2c},x_{2p},x_{1c};y_2)=H(y_2)
  \\
  r_2^c +r_2^p &\leq I(x_{2p},x_{2c};y_2|x_{1c})=H(x_2)
  \\
  r_1^c &\leq I(x_{1c};y_2|x_{2c},x_{2p})=H(x_{1c})
  \\
  r_2^p+r_1^c &\leq I(x_{2p},x_{1c};y_2|x_{2c})=H(y_2|x_{2c})
  \\
  r_2^p&\leq I(x_{2p};y_2|x_{1c},x_{2c})=H(x_{2p})\, .
  \end{split}\end{equation}
 Thus, it suffices to show that each of the mutual information constraints is made looser when using the (optimal) uniform input distribution of the deterministic channel. Note that only the first, fourth, sixth, and ninth constraints are affected by the change to real addition. 
 
 Now, in the deterministic channel \eqref{e:Appendix:detChannel}, the output $y_1$ is uniformly distributed; alternatively, each bit in the binary expansion of $y_1$ that is random is independent of the other bits and has equal probability of being zero or one. The distribution of these bits in the binary expansion of $y_1$ does not change in passing to real addition, because each bit is the sum modulo two of a carry bit and a fresh random bit. 
 It follows that the entropy $H(y_1)$ does not decrease. The entropies $H(y_1|x_{1c})$ and $H(y_2|x_{2c})$ behave similarly. 


\vspace{2mm}
\noindent 
{\bf Step 2: Real-valued gains (lose $\log 3$ bits).} 
In this step we compare the achievable rate under a uniform input distribution of a channel with real-valued gains to the achievable rate in Step 1, losing at most $\log 3$ bits per user. The result is an achievable region that is within $\log 3$ bits per user of the capacity region of the original deterministic channel. 

To allow real-valued gains, we first allow negative cross gains. It is sufficient to consider only the case of cross gains, rather than any of the gains, being negative, since each transmitter can negate its input to ensure a positive signal on the direct link. Viewing each input as coming from a contiguous subset of integers in the real line, it is clear that the entropy constraints in \eqref{e:Appendix:MAC} are invariant to negating a cross gain when the distribution is uniform. 

Next, replace $2^{n_{ij}}$ with the gain $h_{ij}$ having binary expansion 
$$
h_{ij}=\sign(h_{ij})\sum_{k=-n_{ij}}^{\infty}2^{-k} h_{ij}(k)\,.
$$ 
Accordingly, Channel 2 is given by \begin{equation}\begin{split}\label{e:Appendix:Step2}
  y_1&=\left\lf \bigg(\sum_{k=-n_{11}}^{\infty}2^{-k} h_{11}(k)\bigg)\bigg(\sum_{k=1}^{n_{11}}2^{-k} x_1(k)\bigg)\right\rf
  \\ + &\sign(h_{12})\left\lf \bigg(\sum_{k=-n_{12}}^{\infty}2^{-k} h_{12}(k)\bigg)\bigg(\sum_{k=1}^{n_{22}}2^{-k} x_2(k)\bigg)\right\rf\,,\end{split}
\end{equation}  and analogously for $y_2$.
We continue by comparing the mutual information constraints in \eqref{e:Appendix:MAC}, noting that any rate in the intersection of the MACs at each receiver is achievable by coding for the MACs. 
To begin, we may view the first term in \eqref{e:Appendix:Step2} as starting with the random variable $2^{n_{11}}\sum_{k=1}^{n_{11}}2^{-k} x_1(k)$, which is uniformly distributed on $\{0,\dots, 2^{n_{11}}-1\}$, scaled by $\frac{h_{11}}{2^{n_{11}}}\geq 1$, and retaining the integer part $\lf \, \cdot \, \rf$. Upon scaling, any two points in the support are at least distance 1 apart, so the integer part is at least distance 1 as well. Thus, the first term in \eqref{e:Appendix:Step2} is uniformly distributed with support a subset of the integers having cardinality $2^{n_{11}}$; the support now has gaps, and is no longer the set of integers between $0$ and $2^{n_{11}}-1$ (see Figure~\ref{fig:ConstellationStretch}). 

\begin{figure}
\begin{centering}
\psset{unit=1mm,linewidth=.3pt,arrowlength=1.2,
arrowinset=0,labelsep=3pt}
\begin{center}
\begin{pspicture}(5,0)(75,45)

\rput(0,5){
\psline{->}(5,30)(75,30)

\rput(10,36){\scriptsize $\supp( X )$}

\uput[d](10,29){\scriptsize $0$}

\uput[d](50,29){\scriptsize $2^{n}$}

\pscircle[fillstyle=solid,fillcolor=black](10,30){1}
\pscircle[fillstyle=solid,fillcolor=black](15,30){1}
\pscircle[fillstyle=solid,fillcolor=black](20,30){1}
\pscircle[fillstyle=solid,fillcolor=black](25,30){1}
\pscircle[fillstyle=solid,fillcolor=black](30,30){1}
\pscircle[fillstyle=solid,fillcolor=black](35,30){1}
\pscircle[fillstyle=solid,fillcolor=black](40,30){1}
\pscircle[fillstyle=solid,fillcolor=black](45,30){1}

\psline(5,30)(5,32)
\rput(5,0){\psline(5,30)(5,32)}
\rput(10,0){\psline(5,30)(5,32)}
\rput(15,0){\psline(5,30)(5,32)}
\rput(20,0){\psline(5,30)(5,32)}
\rput(25,0){\psline(5,30)(5,32)}
\rput(30,0){\psline(5,30)(5,32)}
\rput(35,0){\psline(5,30)(5,32)}
\rput(40,0){\psline(5,30)(5,32)}
\rput(45,0){\psline(5,30)(5,32)}
\rput(50,0){\psline(5,30)(5,32)}
\rput(55,0){\psline(5,30)(5,32)}
\rput(60,0){\psline(5,30)(5,32)}
\rput(65,0){\psline(5,30)(5,32)}

}

\rput(0,2.5){
\psline{->}(5,15)(75,15)

\rput(10,21){\scriptsize $\supp(h X )$}

\pscircle[fillstyle=solid,fillcolor=black](10,15){1}
\uput[d](10,14){\scriptsize $0$}

\pscircle[fillstyle=solid,fillcolor=black](17,15){1}
\pscircle[fillstyle=solid,fillcolor=black](24,15){1}
\pscircle[fillstyle=solid,fillcolor=black](31,15){1}
\pscircle[fillstyle=solid,fillcolor=black](38,15){1}
\pscircle[fillstyle=solid,fillcolor=black](45,15){1}
\pscircle[fillstyle=solid,fillcolor=black](52,15){1}
\pscircle[fillstyle=solid,fillcolor=black](59,15){1}

\psline(5,15)(5,17)
\rput(5,0){\psline(5,15)(5,17)}
\rput(10,0){\psline(5,15)(5,17)}
\rput(15,0){\psline(5,15)(5,17)}
\rput(20,0){\psline(5,15)(5,17)}
\rput(25,0){\psline(5,15)(5,17)}
\rput(30,0){\psline(5,15)(5,17)}
\rput(35,0){\psline(5,15)(5,17)}
\rput(40,0){\psline(5,15)(5,17)}
\rput(45,0){\psline(5,15)(5,17)}
\rput(50,0){\psline(5,15)(5,17)}
\rput(55,0){\psline(5,15)(5,17)}
\rput(60,0){\psline(5,15)(5,17)}
\rput(65,0){\psline(5,15)(5,17)}
}

\psline{->}(5,0)(75,0)

\rput(10,6){\scriptsize $\supp(\lf h X \rf)$}

\pscircle[fillstyle=solid,fillcolor=black](10,0){1}
\uput[d](10,-1){\scriptsize $0$}

\rput(0,-15){
\pscircle[fillstyle=solid,fillcolor=black](15,15){1}
\pscircle[fillstyle=solid,fillcolor=black](20,15){1}
\pscircle[fillstyle=solid,fillcolor=black](30,15){1}
\pscircle[fillstyle=solid,fillcolor=black](35,15){1}
\pscircle[fillstyle=solid,fillcolor=black](45,15){1}
\pscircle[fillstyle=solid,fillcolor=black](50,15){1}
\pscircle[fillstyle=solid,fillcolor=black](55,15){1}
}

\psline(5,0)(5,2)
\rput(0,-15){
\rput(5,0){\psline(5,15)(5,17)}
\rput(10,0){\psline(5,15)(5,17)}
\rput(15,0){\psline(5,15)(5,17)}
\rput(20,0){\psline(5,15)(5,17)}
\rput(25,0){\psline(5,15)(5,17)}
\rput(30,0){\psline(5,15)(5,17)}
\rput(35,0){\psline(5,15)(5,17)}
\rput(40,0){\psline(5,15)(5,17)}
\rput(45,0){\psline(5,15)(5,17)}
\rput(50,0){\psline(5,15)(5,17)}
\rput(55,0){\psline(5,15)(5,17)}
\rput(60,0){\psline(5,15)(5,17)}
\rput(65,0){\psline(5,15)(5,17)}}

\end{pspicture}
\end{center} 
\caption{Making the gains real-valued creates gaps in the support without changing its cardinality. In this example $n=3$ and $h=1.4 (2^3)=11.2$.}
\label{fig:ConstellationStretch}
\end{centering}
\end{figure}
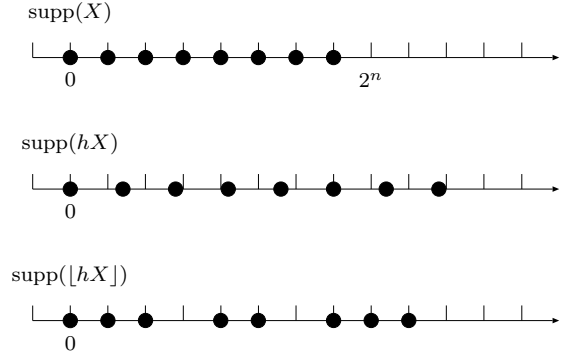

The second term in \eqref{e:Appendix:Step2} is similar, but the the argument must be modified to account for the part of the signal below the noise level. We have
\begin{align}
  &\left\lf \bigg(\sum_{k=-n_{12}}^{\infty}2^{-k} h_{12}(k)\bigg)\bigg(\sum_{k=1}^{\infty}2^{-k} x_2(k)\bigg)\right\rf \nn
  \\&=
  \left\lf |h_{12}| \sum_{k=1}^{n_{12}}2^{-k} x_2(k)+ |h_{12}|\sum_{k=n_{12}+1}^\infty 2^{-k}x_2(k)\right\rf \nn
  \\&:=
  \lf A_1+A_2\rf
  \label{e:Appendix:Step2belowNoiseLevel}
\end{align}
The argument for the first term in \eqref{e:Appendix:Step2} applies to the sum $A_1$ in \eqref{e:Appendix:Step2belowNoiseLevel}, giving that $A_1$ is distributed uniformly with spacing $h_{12}/2^{n_{12}}\geq 1$ and support set having cardinality $2^{n_{12}}$. 
Now, $A_2$ is bounded as $0\leq A_2\leq 2$, since $|h_{12}|\leq 2^{n_{12}+1}$.
Hence, defining $$s=\lf A_1 +A_2\rf- \lf A_1 \rf\, ,$$ we see that $s$ can take on values $0,1,2$, giving 
\begin{equation}\label{e:Appendix:SideInfBoundStep2A}
  H(s)\leq \log 3\,.
\end{equation}
Neglecting $A_2$, let the modified output be  
\begin{equation}\begin{split}\label{e:Appendix:Step2modifiedOutput}
  \y_1&=\left\lf \bigg(\sum_{k=-n_{11}}^{\infty}2^{-k} h_{11}(k)\bigg)\bigg(\sum_{k=1}^{n_{11}}2^{-k} x_1(k)\bigg)\right\rf
  \\ + &\sign(h_{12})\left\lf \bigg(\sum_{k=-n_{12}}^{\infty}2^k h_{12}(k)\bigg)\bigg(\sum_{k=1}^{n_{12}}2^{-k} x_2(k)\bigg)\right\rf\,.\end{split}
\end{equation}
Since $y_1$ can be recovered by the pair $\y_1,s$, Lemma~\ref{l:Appendix:sideInformationLemma} shows that $$I(x_1;\y_1)\geq I(x_1;y_1)-\log 3\,.$$

The argument is completed by using the fact that 
\begin{align*}&H\left( \left\lf |h_{11}|\sum_{k=1}^{n_{11}}2^{-k} x_1(k)\right\rf +\left\lf |h_{12}| \sum_{k=1}^{n_{12}}2^{-k} x_2(k)\right\rf \right)
\\ &\geq
H\left( \sum_{k=1}^{n_{11}} 2^{n_{11}-k}x_1(k)+\sum_{k=1}^{n_{12}} 2^{n_{12}-k}x_2(k)\right)\,.
\end{align*}
This is seen to be true by directly comparing the distributions of the two random variables within the entropies. Counting the number of pairs of integers that sum to each integer, we see that the distribution on the left-hand side can be achieved by shifting probability mass from more likely to less likely values. 

The argument applies to all the mutual information constraints of \eqref{e:Appendix:MAC}.
Step 2 incurs a loss of $\log 3\leq 1.6$ bits.


\vspace{2mm}
\noindent
{\bf Step 3: Additive Gaussian noise (lose $1.5$ bits).} 
Let Channel 3 be obtained from Channel 2 by adding Gaussian noise $z_i\sim \N(0,1)$ to output $i$, where the outputs of Channel 2 are given by \eqref{e:Appendix:Step2modifiedOutput}
\begin{equation}\begin{split}\label{e:noiseQuantizedChannelAppendix}
  y_1=\left\lf |h_{11}|\sum_{k=1}^{n_{11}}2^{-k} x_2(k)\right\rf +\sign(h_{12})\left\lf |h_{12}| \sum_{k=1}^{n_{12}}2^{-k} x_2(k)\right\rf 
\end{split}\end{equation} and similarly for $y_2$. 

   Define the random variable $s = [ z_1]$, where $[\, \cdot\, ]$ is the nearest integer function. Observe that it is possible to recover $y_1^N$ from the pair $(y_1^N+z_1^N,s^N)$.  Lemma~\ref{l:Appendix:sideInformationLemma} gives that $$\onen I(x_1^N;y_1^N+z_1^N) \geq \onen I(x_1^N;y_1^N)-H(s)\, .$$
   It remains only to derive a bound on the entropy of $s$,
   \begin{align*}
   H(s)&=-\sum_{k=-\infty}^\infty \Prob(s=k)\log \Prob(s=k)\\
   &= -2\sum_{k=1}^\infty \Prob(s=k)\log \Prob(s=k)
   \\&\quad -\Prob(s=0)\log P(s=0)\\
   &\leq 1.5\,. 
   \end{align*}
   


\vspace{2mm}
\noindent
{\bf Step 4: Remove truncation at noise level (lose $\log 3$ bits).} Let Channel 4 be the Gaussian channel \eqref{e:GaussianChannelAppendix}  
\begin{align*}
  y_1&= h_{11}x_1+h_{12}x_2 +z_1\\
  y_2&= h_{21}x_1+h_{22}x_2 +z_2\, .
\end{align*}
The difference between Channels 3 and 4 is that signals received below the noise level are no longer truncated at the receivers. The output at receiver 1 is
\begin{align*}y_1 = h_{11}x_1+h_{12}x_2 +z_1
= \y_1 + \x_1+\sign(h_{12})\x_2,
\end{align*}
where $\y_1$ is the output at receiver $1$ in Channel 3 \eqref{e:noiseQuantizedChannelAppendix} and $\x_1,\x_2$ are the magnitudes of the signals received below the noise level at receiver 1. 
    
The approach is similar to Step 3. Define the random variable \begin{equation} \label{e:sideInformationStep3} s = [ \x_1+\sign(h_{12})\x_2]\,\end{equation} where $[\, \cdot\, ]$ is the nearest integer function. 
Each of $\x_1,\x_2$ is bounded between 0 and 1 (since they are below the noise level), and so the random variable $s$ can take at most 3 values. Hence the entropy of $s$ is bounded as $$H(s)\leq \log 3\,.$$
It is possible to recover $\y_1^N$ from the pair $(y_1^N,s^N)$.
Therefore Lemma~\ref{l:Appendix:sideInformationLemma} gives
$$\onen I(x^N;y_1^N)\geq \onen I(x_1^N;\y_1^N)-\log 3\,.$$
This completes the first direction of the proof. 

\begin{remark}
  The above proof used the form of the capacity achieving input distribution. Thus, it does not follow that any capacity achieving distribution for the deterministic channel can simply be used with an outer code in the Gaussian channel. 
\end{remark}

\begin{remark}
  The final achievable strategy uses only positive, peak-power constrained inputs to the channel, which is obviously suboptimal.
\end{remark}


\subsection{$C_{Gaussian}\subseteq C_{det}+(13.6,13.6)$}
\label{subsec:Appendix:GaussianToDet}
Here we begin with the Gaussian channel and finish with the deterministic channel. Most of the steps are precisely the opposite as in the previous section. There is an important difference, however: the inputs to the Gaussian channel satisfy the less stringent average power constraint whereas the inputs to the deterministic channel must satisfy a peak power constraint. An extra step in the argument accounts for this difference.

Step 1 removes the part of the input signals exceeding the peak power constraint; Step 2 truncates the signals at the noise level and removes the noise; Step $2'$ derives a single-letter expression for the capacity region of the channel in Step 2 and shows the near-optimality of uniformly distributed inputs; Step 3 restricts the inputs and channel gains to positive numbers; Step 4 makes addition modulo 2; Step 5 quantizes the channel gains to the form $2^{n_{ij}}$.

Denote by Channel 0 the original Gaussian interference channel,
\begin{equation}\begin{split}\label{e:GaussianChannelAppendix2}
  y_1&=h_{11} x_1+h_{12} x_2+z_1 \\
  y_2&=h_{21} x_2+h_{22} x_2+z_2\,.
\end{split}\end{equation}
Recall that we assumed a unit average power constraint
\begin{equation}\label{e:Appendix:AvgPwrConstraint} \onen\sum_{k=1}^n\E[x_{i,k}^2]\leq 1\,.\end{equation}


\vspace{2mm}
\noindent
{\bf Step 1: Peak power constraint instead of average power constraint (lose 4 bits).} 
The input-output relationship of Channel 1 is the same as Channel 0 \eqref{e:GaussianChannelAppendix2}:
\begin{equation}\label{e:Appendix:PeakPowerConstraint}
  y_i=h_{i1} x_1+h_{i2} x_2+z_i\,.
\end{equation} 
The difference is that the inputs to Channel 1 satisfy a peak power constraint instead of an average power constraint: 
$$x_i\leq 1\,.$$
Writing the binary expansion of $x_i$,
$$x_i=\sum_{k=-\infty}^\infty x_i(k)2^{-k}\,,$$ we see that in Channel 1, $x_i(k)\equiv 0$ for $k\leq 0$. 

Let $x_i$ be an input to Channel 0, satisfying the average power constraint \eqref{e:Appendix:AvgPwrConstraint}. 
Let the part of the input that exceeds the peak power constraint be
\begin{equation*}
  \x_i=\lf x_i \rf = \sign(x_i)\sum_{k=-\infty}^0 x_i(k) 2^{-k}\, ,
\end{equation*}
and let
$$\xb_i=x_i-\x_i=\sign(x_i)\sum_{k=1}^\infty x_i(k)2^{-k}$$ be the remaining signal. The signal $\xb_i$ is defined so as to satisfy the peak power constraint. Finally, denote by $\yb_i$ the output at receiver $i$ when the inputs are truncated to the peak power constraint, 
\begin{align*}
  \yb_i &= h_{i1} \xb_1 +  h_{i2}\xb_2 +z_i \, ,
\end{align*}
and let 
\begin{equation}
  \yh_i = y_i-\yb_i
  = h_{i1}\x_1+h_{i2}\x_2 
   \label{e:Appendix:B:yTilde}
\end{equation} be the output due to the inputs $\x_1,\x_2$. 

To complete Step 1, we show that most of the mutual information $I(x_i^N;y_i^N)$ is preserved when the inputs are truncated to the peak power constraint. 
First, observe that since $x_1$ and $x_2$ are independent, $\x_i^N,\xb_i^N,\yb_i^N$ form a Markov chain, $\x_i^N - \xb_i^N - \yb_i^N$. It follows that 
$$
I(\x_i^N; \yb_i^N|\xb_i^N)=0\,.
$$
 
Hence, from the data processing inequality and the mutual information chain rule we have
\begin{align}
  &I(x_i^N;y_i^N) \nn \\
  &\leq I(\xb_i^N,\x_i^N;\yb_i^N,\yh_i^N)\nn
  \\&=
  I(\xb_i^N,\x_i^N;\yb_i^N)+I(\xb_i^N,\x_i^N;\yh_i^N|\yb_i^N)\nn
  \\&\leq
  I(\xb_i^N;\yb_i^N)+I(\x_i^N;\yb_i^N|\xb_i^N)+H(\yh_i^N)\nn
  \\&=I(\xb_i^N;\yb_i^N)+H(\yh_i^N) \nn
  \\&\leq I(\xb_i^N;\yb_i^N)+H(\x_1^N)+H(\x_2^N)
  \,.\label{e:Appendix:SumEntropiesBound}
\end{align}
The last inequality is a consequence of the fact that $\x_1,\x_2$ determine $\yh_i$.
It remains only to bound each of the entropy terms in \eqref{e:Appendix:SumEntropiesBound}.
\begin{lemma} \label{l:Appendix:PeakPower}The following bound on the entropy holds
  \begin{equation}\label{e:Appendix:entropyBoundExceedingPPC}H(\x_1^N)\leq 2 N  \, .\end{equation}
\end{lemma}
\begin{proof}
  The proof is based on the requirement that the part of $x_i^N$ exceeding the peak power constraint, $\x_i^N$, itself must satisfy the average power constraint. Note that the entropy $H(\x_i^N)$ does not depend on the channel gains at all. The part of the signal satisfying the peak power constraint, $\xb_i$, absorbs all the benefit from increasing the signal to noise ratio, as less significant bits from $\xb_i$ appear above the noise level at the receiver. 
  
  Two approaches are possible. The simpler approach is to observe that any scheme in the point-to-point deterministic channel with average power constraint can be used without modification in the Gaussian channel with power constraint $P=1$, with a loss of at most $1.5$ bits due to noise, by the argument in Step 3 of the previous subsection.  The result then follows from the fact that the capacity of the point-to-point Gaussian channel with average power constraint $P=1$ is $\frac{1}{2}\log (1+1)=\frac{1}{2}$. Thus, $$H(\x_i^N)\leq 2 N\,.$$
  
  Alternatively, one may explicitly bound the number of possible values for $\x_i^N$ using a combinatorial argument. The first step is to notice that for each transmission at power $2^m$, it must hold that $2^m-1$ other time slots are silent. By writing a recursion in $m$ and $N$ on the number of possible signals of length $N$ with peak power between $2^m$ and $2^{m-1}$, it is possible to bound the cardinality of the support of $\x_i^n$ by $\text{poly}(N) c^N$ for a constant $c$ and for all $N$, which shows that $\limsup \onen H(\x_i^N)\leq c$.
\end{proof}

Plugging in the estimate \eqref{e:Appendix:entropyBoundExceedingPPC} from the Lemma into \eqref{e:Appendix:SumEntropiesBound} shows that at most 4 bits per user are lost in passing to a peak power constraint.


\vspace{2mm}
\noindent
{\bf Step 2: Truncate signals at noise level, remove fractional part of channel gains, and remove noise (lose $2.6$ bits).}
The truncation at the noise level is not performed by solely taking the integer part of a real-valued signal; instead, the \emph{binary expansion} of each incoming signal is truncated appropriately, and only then do we take the integer part of each signal. In the final deterministic channel the two procedures are equivalent, so we choose this more convenient option with regards to the proof. The key benefit of this choice of truncation is the resulting clear distinction between common and private information, with the unintended receiver able to decode the common information. The derivation of the single-letter expression for the deterministic channel in Section~\ref{sec:StructureOfRegion} can then be applied without modification in Step $2'$.

We write the peak-power constrained channel inputs as
\begin{equation}\label{e:Appendix:InputBinary}
  x_i=\sign(x_i)\sum_{k=1}^\infty x_i(k) 2^{-k},\quad x_i(k)\in \{0,1\}\,.
\end{equation}

If $\lf \log h\rf = n$, then we deem as being above the noise level the component of $hx$ arising from the $n$ most significant bits in the binary expansion of $x$:
 \begin{equation} \label{e:Appendix:InputAboveNoise}
 h \sign(x)\sum_{k=1}^n 2^{-k} x_i(k)\,.
  \end{equation}
The magnitude of the part below the noise level can be bounded as 
 \begin{equation}\label{e:Appendix:belowNoiseEstimate}
 |h| \sum_{k=n+1}^\infty 2^{-k} x_i(k)\leq 2^{n+1} 2^{-n}=2\,.\end{equation}
Channel 2 is defined by retaining only the part of the inputs above the noise level as described in \eqref{e:Appendix:InputAboveNoise}, taking the integer part of the channel gains, further taking the integer part of each observed signal, and removing the noise. 
%
More specifically, receiver $i$ observes the signal
\begin{equation}\begin{split}\label{e:Appendix:truncatedOutput}
\yb_i = \left\lf \lf h_{i1}\rf \sum_{k=1}^{n_{i1}} 2^{-k} x_1(k)
\right\rf + \left\lf \lf h_{i2}\rf \sum_{k=1}^{n_{i2}} 2^{-k} x_2(k) \right\rf\, .\end{split}
\end{equation} 

Now, denote by $\ve_i$ the difference in the outputs relative to Channel 1, ignoring the additive Gaussian noise: \begin{align*}
  \varepsilon_i:&= y_i-\yb_i    \\
  &= \bigg\{h_{i1} \sign(x_1)\sum_{k=n_{i1}+1}^\infty 2^{-k} x_1(k) \\ &\qquad + (h_{i1}-\lf h_{i1}\rf)\sign(x_1)\sum_{k=1}^{n_{i1}} 2^{-k} x_1(k)
  \\&\qquad + \fracp\left( \sign(x_1) \lf h_{i1}\rf \sum_{k=1}^{n_{i1}} 2^{-k} x_1(k)\right)\bigg\}
  \\
  &\quad + \bigg\{ h_{i2} \sign(x_2)\sum_{k=n_{i2}+1}^\infty 2^{-k} x_2(k) \\&\qquad + (h_{i2}-\lf h_{i2}\rf)\sign(x_2)\sum_{k=1}^{n_{i2}} 2^{-k} x_2(k)
    \\&\qquad+ \fracp\left(\sign(x_2)\lf h_{i1}\rf \sum_{k=1}^{n_{i2}} 2^{-k} x_2(k)\right)\bigg\}+z_i
  \\
  &:= \x_1+\x_2+z_i\,,
\end{align*} where $\fracp(\,\cdot \,)$ denotes the fractional part.
Combining the estimate \eqref{e:Appendix:belowNoiseEstimate} and the fact that $|(h_{ij}-\lf h_{ij}\rf)x_j|\leq 1$, we have
\begin{equation}\label{e:Appendix:ErrorTerm}
|\x_i|\leq 4,\quad i=1,2
  \,.
\end{equation}
We will later use the observation that $\x_1,\x_2\mapsto \ve_i$ forms a Gaussian MAC, and from \eqref{e:Appendix:ErrorTerm} the signal-to-noise ratio is at most 16 for each user.

We show next that $$\onen I(x_i^N;\yb_i^N)+5.1\geq \onen I(x_i^N;y_i^N)\,,$$ where $y_i$ is the output of Channel 1 defined in \eqref{e:Appendix:PeakPowerConstraint}. 
Note that $\yb_i$ is independent of $z_i$. The data processing inequality and the chain rule allow to separate the contribution to the mutual information $I(x_i^N;y_i^N)$ from each term $\ve_i^N,\yb_i^N$:
\begin{align*}
   I(x_i^N;y_i^N)&=I(x_i^N;\yb_i^N+\ve_i^N)
  \\&\leq I(x_i^N;\yb_i^n
  ,\ve_i^N)
  \\&= I(x_i^N;\yb_i^N)+I(x_i^N;\ve_i^N|\yb_i^N)
  \\&\leq I(x_i^N;\yb_i^N)+ I(x_1^N,x_2^N;\ve_i^N | \yb_i^N)
  \\&= I(x_i^N;\yb_i^N)+h(\ve_i^N|\yb_i^N)-h(\ve_i^N|\yb_i^N,x_1^N,x_2^N)
  \\&\leq I(x_i^N;\yb_i^N)+h(\ve_i^N)-h(\ve_i^N|\yb_i^N,x_1^N,x_2^N)
  \\&=I(x_i^N;\yb^N_i)+h(\ve_i^N)-h(z_i^N)
  \\&= I(x_i^N;\yb_i^N)+I(\x_1^N,\x_2^N;\ve_i^N)
  \\&\leq I(x_i^N;\yb_i^N)+2.6 N \, ,
\end{align*} where the last inequality holds for sufficiently large $N$. 
In the last step we used the fact that $\x_1,\x_2\mapsto \ve_i$ forms a Gaussian MAC with signal-to-noise ratio at most 16 for each transmitter, so $\onen I(\x_1,\x_2;\ve_i)\leq \frac{1}{2}\log(1+2(16))+\eps_N$ (with $\eps_N\to 0$).
This completes Step 2. 


\vspace{2mm}
\noindent
{\bf Step $2'$: Single letter expression and near optimality of uniform input distribution (lose $2$ bits).} We now show that the derivation of Section~\ref{sec:StructureOfRegion}, giving a single letter expression for the capacity region of the deterministic channel \eqref{e:MACwithTimeSharing}, applies to the channel of Step 2. Following this, we will prove that using uniformly distributed inputs incurs a loss of at most two bits per user relative to the optimal input distribution. 

Define 
\begin{equation}\label{e:Appendix:CommonInfBits}
x_{2c}:=\sign(x_2)\sum_{k=1}^{n_{12}} 2^{-k}x_2(k)\,,\end{equation}
and similarly for $x_{1c}$. This is the part of the input that causes interference at the unintended receiver. 
Consider the signal that remains at receiver 1 after successfully decoding and subtracting off  $x_1$. From \eqref{e:Appendix:truncatedOutput}, the remaining signal is 
\begin{equation}
  f(x_{2c}):=\left\lf \lf h_{12}\rf x_{2c}\right \rf
  = \left\lf \sign(x_2)\lf h_{12}\rf \sum_{k=1}^{n_{12}} 2^{-k}x_2(k)\right \rf\,.
\end{equation}
The statement that $f:\supp(x_{2c})\to \Z$ is injective is equivalent to the claim that receiver 1 can recover $x_{2c}$ from $f(x_{2c})$.
Now, viewed as a real number, the support of $x_{2c}$ has a spacing of $2^{-n_{12}}$, and since 
\begin{equation}\label{e:Appendix:spacingAtLeastOne}\lf h_{12}\rf\geq 2^{n_{12}}\,,\end{equation}
the spacing of the support of $\lf h_{12}\rf x_{2c}$ is greater than 1. Hence the integer part $\lf\, \cdot\,  \rf$ sends two different values of $\lf h_{12}\rf x_{2c}$ to two different integers, i.e. $f$ is injective. An analogous argument shows that receiver 2 can recover $x_{1c}$.

Since each receiver can recover the common portion of the interfering signal \eqref{e:Appendix:CommonInfBits}, the arguments of Lemmas~\ref{l:separateCommonPrivate} and~\ref{l:MACintersection} in Section~\ref{sec:StructureOfRegion} apply without modification to the channel under scrutiny. Thus, the region is given by \eqref{e:Appendix:MAC}. 

We now show that at most one bit per user is lost relative to the capacity region when each of the signals $x_{1c},x_{1p},x_{2c},x_{2p}$ is uniformly distributed on its support. We first prove a comparable result for random variables with support sets that are arithmetic progressions of integers.
\begin{lemma}\label{l:Appendix:arithmeticProgLemma}
Let $A,B\in \Z$ be two arithmetic progressions,
\begin{align*}
  A&=\{0,a,2a,\dots,(M_A-1) a\}=[0,M_A-1]\cdot a\\ 
  B&=\{0,b,2b,\dots,(M_B-1) b\}=[0,M_B-1]\cdot b\,.
\end{align*} If $X$ and $Y$ are independent and distributed uniformly on $A$ and $B$, respectively, then
\begin{equation}
  H(X+Y)+1 \geq H(X^*+Y^*) 
\end{equation} for any random variables $X^*,Y^*$ with support sets $A,B$. 
\end{lemma}
\begin{proof}
  Scaling the sets $A$ and $B$ by the same number does not change the relevant entropies, so we may assume without loss of generality that $\gcd(a,b)=1$.
  We first estimate the cardinality of the sumset $A+B=\{a+b:a\in A,b\in B\}$. 
  Note that $$A+B\subseteq \{0,\dots,a(M_A-1)+b(M_B-1)\}\,,$$ from which it follows that
  \begin{equation}\label{e:Appendix:sumsetCardinalityBound}
    |A+B|\leq  a M_A+b M_B\,.
  \end{equation}
  Since $\supp(X^*+Y^*)\subseteq A+B$, we therefore have the estimate 
  \begin{equation}\label{e:Appendix:maxSumBound}
    H(X^*+Y^*)\leq \log (a M_A+b M_B)\,.
  \end{equation}
  
  Next we calculate the maximum probability mass in the distribution of $X+Y$,
  \begin{equation}\label{e:Appendix:maximumMass}\bar{p}:=\max_{x\in A+B} \Prob(X+Y=x)\,.\end{equation}
  For each $k$ with $0\leq k\leq M_B-1$ let $$S_k:=A+kb=[0,M_A-1]\cdot a+kb\,.$$ 
  For $k$ outside the interval $[0,M_B-1]$, $S_k$ is defined to be empty. A typical element of $S_k\cap S_{k'}$ with $k'\leq k$ can be written as $$q a+k b = q' a+k' b,$$ for some $0\leq q\leq M_A-1$ and $0\leq q'\leq M_B-1$. Rearranging, we have $$(k-k') b = (q'-q)a\,,$$ which by the assumption $\gcd(a,b)=1$ implies $$a|(k-k')\,. \quad 
  $$ 
  Thus \begin{equation}\label{e:Appendix:intersectionCondition}S_k\cap S_{k'}\neq \emptyset \quad \text{implies} \quad k\equiv k'\mod a\,.\end{equation}
  Letting $\A$ and $\B$ be shifts of $A$ and $B$ so that a median point lies at the origin, the maximum in \eqref{e:Appendix:maximumMass} occurs at $x=0$, and it can be seen from the condition \eqref{e:Appendix:intersectionCondition} that 
  \begin{align*}
  |\{x,y:x+y=0,x\in A,y\in B\}|\leq \min\left(\frac{M_A}{b},\frac{M_B}{a}\right)\,.
  \end{align*}
  Since for each $x\in A,y\in B$, $P(X=x)=1/M_A$ and $P(Y=y)=1/M_B$, and $X$ and $Y$ are independent, 
  \begin{align*}
  -\log \bar{p}&=-\log \sum_{x\in A,y\in B\atop x+y=0} P(X=x,Y=y)
  \\
  &=-\log \frac{|\{x,y:x+y=0,x\in A,y\in B\}|}{M_A M_B}\\
  &\geq \log \frac{M_A M_B}{\min(\frac{M_A}{b},\frac{M_B}{a})}\\
  &= \max (\log (a M_A),\log(bM_B))\,.
  \end{align*}
  Hence, from equation \eqref{e:Appendix:maxSumBound}, \begin{equation}\begin{split}\label{e:Appendix:ProgressionEntropyComparison}
  H(X+Y)&=-\sum_{x\in A+B}p(x)\log p(x)
  \\ &\geq
  -\sum_{x\in A+B}p(x)\log \bar{p}
  \\ &\geq \max (\log (a M_A),\log(bM_B))
  \\ &\geq  \log(a M_A+b M_B)-1
  \\ &\geq  H(X^*+Y^*)-1\,.\end{split}
  \end{equation} 
  This proves the lemma.
\end{proof}

It is not difficult to extend the proof of the Lemma to show the near optimality of uniformly distributed inputs for the channel defined by \eqref{e:Appendix:truncatedOutput}. Let
\begin{equation}
  U:= \lf h_{i1}\rf \sum_{k=1}^{n_{i1}} 2^{-k} x_1(k)
 \end{equation}
 and
 \begin{equation}
 V:=  \lf h_{i2}\rf \sum_{k=1}^{n_{i2}} 2^{-k} x_2(k) \,,
\end{equation}
so that $$y=\lf U \rf+\lf V\rf\,.$$
Also, let \begin{align*}A:&=\supp(U)=\{0,\lf h_{i1}\rf,\dots,\lf h_{i1}\rf (2^{n_{i1}}-1)\}\cdot 2^{-n_{i1}},
\\ B:&=\supp(V) =\{0,\lf h_{i2}\rf,\dots,\lf h_{i2}\rf (2^{n_{i2}}-1)\}\cdot 2^{-n_{i2}}\,.\end{align*}

Assume without loss of generality (by symmetry of the definitions of $U$ and $V$) that $n_{i1}\geq n_{i2}$. We will work with scaled, integer-valued versions of $U$ and $V$: let $$\De:=2^{n_{i1}}\,$$ and  
$$\U:=\De U,\quad \V:=\De V\,.$$ Let $M_A=\De$ and $M_B=2^{n_{i2}}$. The supports sets are
\begin{equation*}
   \A=\{0,1,\dots,(M_A -1)\}\cdot \lf h_{i1}\rf
 \end{equation*}
 and
\begin{equation*}
  \B=\{0,1,\dots,(M_B-1)\}\cdot \De (\lf h_{i2}\rf 2^{-n_{i2}})\,.
\end{equation*}
 Correspondingly, the integer part of a number $t$ is replaced by quantization to the greatest multiple of $\De$ less than or equal to $t$:
$$Q(t):=\De \left\lf \frac{t}{\De}\right\rf\,.$$
In the notation of Lemma~\ref{l:Appendix:arithmeticProgLemma}, the spacings in the sets $\A$ and $\B$ are, respectively, $a=\lf h_{i1}\rf$ and $b=\De (\lf h_{i2}\rf 2^{-n_{i2}})$.
Proving the equivalent of Lemma~\ref{l:Appendix:arithmeticProgLemma} for $Q(\U)+Q(\V)$ will imply the same result for $y=\lf U\rf+\lf V\rf$ by the scale-invariance of discrete entropy.

With this notation, we have analogously to \eqref{e:Appendix:sumsetCardinalityBound} that
\begin{equation}\label{e:Appendix:cardinalityBoundQuantized}|Q(\A)+Q(\B)|\leq \frac{aM_A+bM_B}{\De}\,.\end{equation}
The next step is to compute a bound on the maximum probability mass in $Q(\U)+Q( \V)$,
$$p^*:=\max_x \Prob(Q(\U) +Q(\V)=x)\,.$$
For any $x$, we have
\begin{align*}
  \{u\in \U,v\in \V: Q(u)+Q(v)=x\} 
  &\subseteq 
  \{u\in \U,v\in \V:  u  +  v\in [x,x+2\De) \} \\
  &=
  \bigcup_{x^*\in [x,x+2\De)}  \{u\in \U,v\in \V:  u  +  v= x^* \}\,.
\end{align*}
Thus 
\begin{equation}
\begin{split}\label{e:Appendix:maximumMassQuantized}
  p^* &\leq \max_x \sum_{x^*\in [x,x+2\De)} \Prob(\U+\V=x^*)
  \\ &\leq
  2 \De \bar{p}\,,
\end{split}\end{equation} 
where $\bar{p}$ is defined in \eqref{e:Appendix:maximumMass}.
Combining \eqref{e:Appendix:cardinalityBoundQuantized} and \eqref{e:Appendix:maximumMassQuantized}, the desired result now follows exactly as in equation \eqref{e:Appendix:ProgressionEntropyComparison} of Lemma~\ref{l:Appendix:arithmeticProgLemma}, giving that 
$$H(\U+\V)\geq H(\U^*+\V^*)-2\,.$$ 

The near optimality of the uniform distribution applies to each entropy constraint in \eqref{e:Appendix:MAC}, and thus each user loses at most 2 bits as claimed.


\vspace{2mm}
\noindent
{\bf Step 3: Positive inputs and channel gains (lose $2$ bits).} From Step $2'$, the uniform distribution is nearly optimal for Channel 2.
Viewing the inputs as coming from a constellation in the real line, it is not hard to see that negating a cross gain does not change any of the output statistics, therefore preserving the mutual information. Similarly, each of the output entropies in \eqref{e:Appendix:MAC} is reduced by at most 2 bits if the inputs are restricted to be positive.

\vspace{2mm}
\noindent
{\bf Step 4: Addition over $\F_2$ (lose $2$ bits).}  
Consider the binary expansion of the output. In switching to modulo 2 addition, every output bit that has some entropy when using real addition is uniformly random, except possibly the two most significant bits that arise due to carry-overs. Thus, at most two bits are lost in each of the entropy constraints of \eqref{e:Appendix:MAC}.

\vspace{2mm}
\noindent
{\bf Step 5: Channel gains of the form $2^n$ (lose zero bits).}
Channel 5 is the deterministic channel \eqref{e:Appendix:detChannel}. 
The optimal input distribution is uniform and the mutual information is unchanged when the gains are quantized to the nearest power of 2. In fact, the capacities of the channel in Step 4 and the channel of Step 5 are identical.

\subsection{Complex Gaussian IC}
The proof of Theorem~\ref{t:DeterministicApproximation} in the generality of complex-valued gains and signals is very similar to the proof of Theorem~\ref{t:Appendix:DetApproxReal} for the real-valued channel presented in Sections~\ref{subsec:Appendix:DetToGaussian} and \ref{subsec:Appendix:GaussianToDet}. 
We focus on the proof that $$\CalC_{\text{Gaussian}}\subseteq \CalC_{\text{det}}+\text{ constant}\,;$$ the other direction follows by reversing the steps and using the argument for the real-valued channel, and is omitted. The eventual gap is $42$ bits, roughly double that of the real-valued case. 

The complex Gaussian interference channel is given by 
\begin{align*}
  y_1&=h_{11}x_1+h_{12}x_2+z_1\\
  y_2&=h_{21}x_1+h_{22}x_2+z_2\,,
\end{align*}
where $z_i\sim \CN(0,1)$ and the channel inputs satisfy an average power constraint
$$\onen \sum_{k=1}^N \E[x_{i,k}^2]\leq P_i,\quad i=1,2\,.$$
By scaling the outputs, we may set $P_i=2$ and $z_i\sim \CN(0,2)$.
We assume without loss of generality that the cross gains have zero phase, i.e. $\text{Im}(h_{12})=\text{Im}(h_{21})=0$, since each of the receivers may simply rotate the output appropriately. 
These assumptions allow to write the output of the channel as
\begin{equation}\begin{split}
  \begin{pmatrix} y_{1R} \\ y_{1I} \end{pmatrix} &=
  \begin{pmatrix}
    h_{11}^R & -h_{11}^I \\ h_{11}^I & h_{11}^R
  \end{pmatrix}
  \begin{pmatrix}
    x_{1R} \\ x_{1I}
  \end{pmatrix} +
  \begin{pmatrix}
    h_{12}^R & 0 \\ 0 & h_{12}^R
  \end{pmatrix}
  \begin{pmatrix}
    x_{2R} \\ x_{2I}
  \end{pmatrix}
  +
  \begin{pmatrix}
    z_{1R}\\z_{1I}
  \end{pmatrix}\,,
  \end{split}
\end{equation} and similarly for $y_2$.
Here $R$ and $I$ denote real and imaginary part, respectively, and $z_{iR},z_{iI}\sim \N(0,1)$.

\vspace{2mm}
\noindent
{\bf Step 1: Peak power constraint instead of average power constraint (lose 8 bits).} 
The argument is almost identical to that of Step 1 in \ref{subsec:Appendix:GaussianToDet}. We truncate the inputs, letting the part of the input $x_{iR}$ that exceeds the peak power constraint be
\begin{equation*}
  \x_{iR}=\lf x_{iR} \rf = \sign(x_{iR})\sum_{k=-\infty}^0 x_{iR}(k) 2^{-k}\, ,
\end{equation*}
and let
$$\xb_{iR}=x_{iR}-\x_{iR}=\sign(x_{iR})\sum_{k=1}^\infty x_{iR}(k)2^{-k}$$ be the remaining signal, with similar definitions for $x_{iI}$ with $I$ replacing $R$. The signals $\xb_{iR},\xb_{iI}$ are defined so that $\xb_i=\xb_{iR}+j\xb_{iI}$ satisfies the peak power constraint of $2$.
Let $\yb_i$ be the output at receiver $i$ due to the truncated inputs. The development in Step 1 of \ref{subsec:Appendix:GaussianToDet} shows that 
\begin{align}
  I(x_i^N;y_i^N) \leq I(\xb_i^N;\yb_i^N)+H(\x_1^N)+H(\x_2^N)
  \,.\label{e:Appendix:ComplexSumEntropiesBound}
\end{align}

The estimate $$H(\x_i^N)\leq 4 N$$ follows from the argument of Lemma~\ref{l:Appendix:PeakPower}, by translating an arbitrary strategy for a point-to-point deterministic channel to a corresponding Gaussian channel with $\snr=1$, with a loss of at most 3 bits (1.5 bits per complex dimension). The point-to-point Gaussian channel has capacity 1, giving the estimate.

\vspace{2mm}
\noindent
{\bf Step 2: Truncate signals at noise level, remove fractional part of channel gains, and remove noise (lose $5.1$ bits).}
The argument repeats that of Step 2 in \ref{subsec:Appendix:GaussianToDet}, and is omitted. 

\vspace{2mm}
\noindent
{\bf Step $2'$: Single letter expression, decoupling of real and imaginary components, and near optimality of uniform input distribution (lose $6$ bits).}
After decoding the message of the intended user, each receiver has a clear view of the common message of the interfering user. Thus, the capacity region of the channel of Step 2 is given by \eqref{e:Appendix:MAC}. 

Next, using a similar argument to that in Step $2'$ for the real-valued case, it can be shown that i.i.d. uniformly distributed inputs are nearly optimal on a modified channel, with a loss of at most $4$ bits per user. The modified channel replaces the direct gain $h_{ii}^R$ with $|h_{ii}^R|+|h_{ii}^I|$, and sets $h_{ii}^I=0$. The support of the output is at least as large in the modified channel under uniformly distributed inputs, and moreover, the output is independent over time. Thus, this step decouples the real and imaginary components. The argument for the real-valued channel can now be applied to the real and imaginary components of the complex channel.

\vspace{2mm}
\noindent
{\bf Steps 3, 4, and 5: Positive inputs and channel gains (lose 4 bits), addition over $\F_2$ (lose 2 bits), channel gains of the form $2^n$.}
Steps 3 and 4 are identical to the real-valued case. In Step 5 the direct gains $|h_{ii}^R|+|h_{ii}^I|$ are replaced with $2^{\lf \log(|h_{ii}^R|+|h_{ii}^I|)\rf}$. 
Similarly, the cross gains $|h_{12}^R|$ and $|h_{21}^R|$ are replaced with $2^{\lf \log |h_{12}^R|\rf}$ and $2^{\lf \log |h_{21}^R|\rf}$, respectively.

\vspace{2mm}
\noindent
{\bf Step 6: Combine real and imaginary parallel channels (lose 4 bits).}
Now, the resulting deterministic channel from Step 5 is precisely the same as the deterministic channel in the real-valued case, but with twice as many channel uses (one each for the real and imaginary part of the signal). Hence the capacity region of the complex deterministic channel is the same as for the real-valued channel, but scaled by two. Note that the capacity region for the deterministic channel \eqref{e:explicitMACregion} exactly doubles when all the channel gains are squared. We have
\begin{align*}2^{2\lf \log(|h_{ii}^R|+|h_{ii}^I|)\rf} \leq 2^{\lf 1+\log (|h_{ii}^R|^2+|h_{ii}^I|^2)\rf} =
2^{1+\lf \log \snr_i\rf}\,,\end{align*} which shows that changing the gain to $2^{\lf \log \snr_i\rf}$ changes at most one bit of the output in each complex dimension. Similarly, at most one bit of the output at receiver 1 is changed by changing the cross gain $2^{2\lf \log |h_{12}^R|\rf}$ to $2^{\lf \log \inr_2\rf}$. Thus, at most 4 bits per user are lost in making this final modification to the channel.

\bibliographystyle{plain}
\bibliography{BIBFILE}

\end{document}